\documentclass{article}

\usepackage[affil-it]{authblk}
\usepackage[usenames,dvipsnames]{xcolor}
\usepackage{amsfonts}
\usepackage{amsmath,amsthm,amssymb,dsfont}
\usepackage{enumerate}
\usepackage{graphicx}	
\usepackage{subcaption}
\usepackage[margin=3cm]{geometry}
\usepackage{url}
\usepackage{todonotes}
\usepackage{bbm}

\usepackage{tikz}

\usepackage{pifont}
\usepackage{multirow}
\usepackage{makecell}

\usepackage{epsfig}
\usetikzlibrary{shapes.symbols,patterns} % for source symbols
\usepackage{pgfplots}
\pgfplotsset{compat=1.10}
\usepgfplotslibrary{fillbetween}

\definecolor{linkblue}{HTML}{001487}
\usepackage{hyperref}
\hypersetup{colorlinks=true,citecolor=linkblue,linkcolor=linkblue,filecolor=linkblue,urlcolor=linkblue,breaklinks=true}

\usepackage{nicefrac}
\usepackage{mathtools}
\usepackage[nameinlink,capitalize,noabbrev]{cleveref}

\usepackage{algorithm}
\usepackage{algorithmic}

\usepackage{mdframed}
\usepackage{aligned-overset}
\usepackage{dsfont}
\usepackage{qcircuit}

\theoremstyle{plain}
\newtheorem{theorem}{Theorem}[section]
\newtheorem{lemma}[theorem]{Lemma}

\crefname{question}{Question}{Questions}

\theoremstyle{definition}

\newtheorem{remark}[theorem]{Remark}

\newcommand*{\ee}{\mathrm{e}}

\newcommand*{\cH}{\mathcal{H}}

\newcommand*{\cI}{\mathcal{I}}

\newcommand*{\cO}{O}

\newcommand*{\cR}{\mathcal{R}}
\newcommand*{\cS}{\mathcal{S}}
\newcommand*{\cT}{\mathcal{T}}
\newcommand*{\cU}{\mathcal{U}}
\newcommand*{\cV}{\mathcal{V}}

\newcommand*{\N}{\mathbb{N}}

\newcommand*{\C}{\mathbb{C}}

\newcommand*{\St}{\mathrm{S}}

\newcommand*{\dist}{\mathrm{dist}}

\newcommand*{\eps}{\varepsilon}

\newcommand*{\id}{\mathds{1}}
\newcommand*{\poly}{\mathrm{poly}}

\newcommand*{\tr}{\mathrm{tr}}
\newcommand*{\ket}[1]{| #1 \rangle}
\newcommand*{\bra}[1]{\langle #1 |}
\newcommand*{\spr}[2]{\langle #1 | #2 \rangle}

\newcommand*{\ci}{\mathrm{i}} % imaginary number
\newcommand{\norm}[1]{\left\lVert#1\right\rVert}

\newcommand{\epsjodd}{\eps_{j_{\textnormal{odd}}}}
\newcommand{\epsjeven}{\eps_{j_{\textnormal{even}}}}

\DeclareRobustCommand{\abbrevcrefs}{%
\Crefname{theorem}{Thm.}{Thms.}%
\Crefname{corollary}{Cor.}{Cors.}%
\Crefname{lemma}{Lem.}{Lems.}%
\Crefname{proposition}{Prop.}{Props.}%
\Crefname{equation}{Eq.}{Eqs.}%
\Crefname{example}{Ex.}{Exs.}%
}

\DeclareRobustCommand{\Cshref}[1]{{\abbrevcrefs\Cref{#1}}}

\usepackage{arydshln}
\usepackage{bbold}

\usepackage{float}

 \allowdisplaybreaks

\title{Approximate Quantum Fourier Transform \\in Logarithmic Depth on a Line}
\author{Elisa~B\"aumer\thanks{eba@zurich.ibm.com}\phantom{,}}
\author{David~Sutter}
\author{Stefan~Woerner}
\affil{\small IBM Quantum, IBM Research Europe -- Zurich}
\date{}

\begin{document}

\maketitle

\vspace{-8mm}

\begin{abstract}
The approximate quantum Fourier transform (AQFT) on $n$ qubits can be implemented in logarithmic depth using $8n$ qubits with all-to-all connectivity, as shown in [Hales, PhD Thesis Berkeley, 2002]. However, realizing the required all-to-all connectivity can be challenging in practice. In this work, we use dynamic circuits, i.e., mid-circuit measurements and feed-forward operations, to implement the AQFT in logarithmic depth using only $4n$ qubits arranged on a line with nearest-neighbor connectivity. 
Furthermore, for states with a specific structure, the number of qubits can be further reduced to $2n$ while keeping the logarithmic depth and line connectivity. 
As part of our construction, we introduce a new implementation of an adder with logarithmic depth on a line, which allows us to improve the AQFT construction of Hales.
\end{abstract}

%\tableofcontents

\section{Introduction} \label{sec_intro}
It is widely believed that quantum computers can solve problems that no classical algorithm can solve in a reasonable amount of time. Arguably, the most prominent example is Shor's algorithm~\cite{shor_algo} to factor an arbitrary integer in polynomial time, whereas the most efficient known classical factoring algorithms have a sub-exponential runtime~\cite{factoring_book}. 

To quantify the cost of implementing a quantum algorithm represented by a quantum circuit, typically three quantities are considered: (i) the width, i.e., the number of qubits, (ii) the depth, i.e., the number of consecutive time steps necessary to execute the circuit, where a time step is a single application of the maximum number of gates such that all applied gates can be executed simultaneously, and (iii) the size, i.e., the total number of gates.
Width and depth are arguably the most operationally relevant metrics, as they quantifiy the required number of qubits and resulting runtime of an algorithm.
While occasionally only the size and width of an algorithm are analyzed, this only gives a lower bound on the depth. In contrast, analyzing the width and depth directly results in an upper bound on the size.
Here, we assume that gates on distinct subsets of qubits that are applied in the same time slice can be executed in parallel.

One crucial building block in Shor's algorithm is the quantum Fourier transform (QFT). In~\cite{clevewatrous2000}, Cleve and Watrous showed that the approximate quantum Fourier transform (AQFT) on $n$ qubits can be implemented in depth $\cO(\log n)$ with all-to-all connectivity between the qubits and $\cO(n\log n)$ ancilla qubits. They also reduced the circuit size of the exact QFT from $\cO(n^2)$ to $\cO(n(\log n)^2 \log \log n)$, still with a depth of $\cO(n)$ though. Later in~\cite{hales2002}, Hales showed that the AQFT can be implemented in the same depth, however, requiring only $8n$ qubits. There have been further improvements regarding the size of the AQFT but with linear depth~\cite{Ahokas2004}. The depth can be further reduced to constant when using $\cO(n^3\log n)$ qubits and all-to-all connectivity~\cite{hoyer2005,buhrman2023state}. 
Here, we focus on logarithmic depth constructions that require only $\cO(n)$ qubits on a line.

A natural question is whether the width or connectivity requirements to implement the AQFT can be further improved while preserving logarithmic depth. In this work, we show that 
\begin{enumerate}[(1)]
    \item we can relax the connectivity requirements from all-to-all to a minimum, i.e., to a line of qubits with nearest-neighbor connectivity (see~\cref{tab_QFT_depth}).
    \item we can reduce the number of qubits to $4n$ for a line of qubits with nearest-neighbor connectivity and to 
    $3n$ for an all-to-all connectivity  (see~\cref{tab_QFT_width}).
\end{enumerate}
We refer to~\cref{thm_QFT_1D} for the mathematically rigorous statement.
Our algorithm builds on the results from~\cite{hales2002} and we use the power of dynamic circuits, i.e., mid-circuit measurements and feed-forward operations, to implement the circuits on a line without increasing the required depth. 
In particular, we show how to implement a quantum adder in logarithmic depth on a line of qubits using dynamic circuits, enabling our main result.
For states with specific structure, we do not even require dynamic circuits and can implement the AQFT on a line of only $2n$ qubits.

\subsection{Main results}
Let $N:=2^n$, where $n \in \mathbb{N}$ is the number of qubits. Furthermore, let $\omega_N:=\ee^{\frac{2\pi \ci}{N}}$ be the $N$-the root of unity. 
We denote the quantum Fourier state corresponding to the computational basis state $\ket{j}$, $j \in \{0, \ldots, N{-}1\}$, by
\begin{align}
\ket{\phi(j)} := \frac{1}{\sqrt{N}} \sum_{k=0}^{N-1} \omega_N^{j k} \ket{k} \, .
\end{align}
% and for simplicity we denote $\ket{a}\otimes \ket{b}$ by $\ket{a}\ket{b}$, where each register corresponds to $n$ qubits.
Note that $\ket{\phi(0)} = \frac{1}{\sqrt{N}} \sum_{k=0}^{N-1} \ket{k} = H^{\otimes n} \ket{0}$, i.e., we can switch between $\ket{\phi(0)}$ and $\ket{0}$ by applying a layer of Hadamard gates.
The QFT mapping performs the operation
\begin{align} \label{eq_QFT_def}
   \ket{j}_A \ket{0}_B
    \overset{\textnormal{QFT}_{AB}}{\longrightarrow}  \ket{\phi(j)}_A \ket{0}_B    \, ,
\end{align}
where $A$ denotes the $n$-qubit target register and $B$ is an $n$-qubit ancilla register that is required by the construction used below.
Note that the resulting order of qubits in $A$ is reversed.

The AQFT is a unitary that is close to the exact QFT. More precisely, 
for $\eps >0$, an AQFT acting on an $n$-qubit register $A$, as well as an $n$-qubit ancilla register $B$, is defined as an isometry \smash{$\mathrm{QFT}^{(\eps)}_{AB}$}, such that
\begin{align}
\dist_{\cS}(\mathrm{QFT}_{AB},\mathrm{QFT}^{(\eps)}_{AB})
\leq \eps \, ,  \label{eq:AQFT_def}
\end{align}
where $\dist_{\cS}(\cdot,\cdot)$ will be defined formally in~\Cref{sec:notation}, $\mathrm{QFT}_{AB}$ denotes the quantum Fourier transform unitary defined in~\cref{eq_QFT_def} and $\cS$ is the set of arbitrary states in $A$ and $E$ with $\dim (E) = M \geq 2^{2n}$ and $\ket{0}$ on an $n$-qubit ancilla register $B$, i.e.,
\begin{align}
 \cS:=\Big\{\ket{\psi}_{ABE} \in \cH_{ABE}: \ket{\psi}_{ABE}=\sum_{j=0}^{N-1}\sum_{m=0}^{M-1} \alpha_{j,m} \ket{j}_A \ket{0}_B \ket{m}_{E} \Big\}    \, ,\label{eq_set_S}
\end{align}
where $\ket{\psi}_{ABE} \in \cH_{ABE}$ denotes the set of normalized, pure quantum states on a tripartite Hilbert space $\cH_{ABE}$.
Before explaining the AQFT operation in full generality, we consider a restricted version thereof that assumes input states that have a certain structure. More precisely, we focus on input states from the following set 
\begin{align} \label{eq_S_uni}
 \cS_{\mathrm{uni}}^{(p)}:=\Big\{\ket{\psi}_{ABE} \in \cH_{ABE}: \ket{\psi}_{ABE}=\sum_{j=0}^{N-1}\sum_{m=0}^{M-1} \beta_{j,m} \ket{j}_A \ket{0}_B \ket{m}_{E}  \, , |\sum_{m=0}^{M-1} \beta_{j,m}\beta_{\ell,m}^\ast|\leq\frac{p(n)}{N}\delta_{j,\ell} \; \forall j,\ell \Big\} \, ,
\end{align}
with some environment $E$ with dimension $\dim (E) = M \geq 2^{2n}$, where $\delta_{j,\ell}$ denotes a Kronecker delta, and $p(n)$ denotes a polynomial in $n$ with fixed degree independent of $n$, i.e., states with amplitudes that are somewhat uniformly distributed in $A$. If we choose the polynomial $p(n)$ in~\cref{eq_S_uni} to be the constant $1$ we obtain the set $\cS_{\mathrm{uni}}^{(1)}$.  
Note that clearly $\cS_{\mathrm{uni}}^{(p)} \subseteq \cS$.
For $\eps >0$ the AQFT acting on states in $\cS_{\mathrm{uni}}^{(p)}$, is defined as a unitary \smash{$\mathrm{QFT}_{\mathrm{uni}}^{(\eps)}$}, such that
\begin{align}
\dist_{\cS_{\mathrm{uni}}^{(p)}}(\mathrm{QFT},\mathrm{QFT}_{\mathrm{uni}}^{(\eps)})
\leq \eps \, .  \label{eq:AQFT_uni_def}
\end{align}

\begin{lemma}[AQFT for uniform inputs] \label{lem_QFT_1D}
Let $n \in \N$, \smash{$\frac{1}{\poly(n)}  \leq \eps < 1$} and \smash{$\cS_{\mathrm{uni}}^{(p)}$} as defined in~\cref{eq_S_uni}. A unitary \smash{$\mathrm{QFT}^{(\eps)}_{\mathrm{uni}}$} acting on \smash{$\cS_{\mathrm{uni}}^{(p)}$} satisfying~\cref{eq:AQFT_uni_def} can be implemented on a line of $2n$ qubits with nearest-neighbor connectivity with depth \smash{$\cO(\log \frac{n p(n)}{\eps^2})$}. Specifically, if acting on states in \smash{$\cS_{\mathrm{uni}}^{(1)}$}, it can be implemented in depth $\cO(\log \frac{n}{\eps^2})$. 
\end{lemma}
The proof is given in~\cref{sec:uniform}. While it is guaranteed to work for uniform inputs of the form~\cref{eq_S_uni}, it also works for most random input states $\ket{\psi}_{ABE}\in \cS$ (see \cref{remark:randomstate}). Thus, implementing \smash{$\mathrm{QFT}^{(\eps)}_{\mathrm{uni}}$} may be very relevant in practice, as it can always be the first try if the result is verifiable, such as for Shor's algorithm (see~\cref{sec:shors} for more details).
We want to emphasize that our proof uses the ideas introduced in~\cite{hales2002}, maps the algorithm to a line of qubits and reduces its width and depth by constant factors.

The \smash{$\mathrm{QFT}^{(\eps)}_{\mathrm{uni}}$} implementation consists of two elementary operations that will be approximated:
\begin{enumerate}[(a)]
    \item Quantum Fourier state computation (QFS): $\ket{j}\ket{\phi(b)} \rightarrow \ket{j}\ket{\phi(b+j)}$
    \item Fourier phase estimation (FPE): $\ket{b}\ket{\phi(j)}\rightarrow \ket{b\oplus j}\ket{\phi(j)}$\footnote{$b\oplus j$ denotes the bit-wise addition modulo $2$ (XOR).}
\end{enumerate}
together with a Hadamard transform and $\mathrm{SWAP}$ gates. 
As illustrated in the circuit in~\cref{fig:circuit_uni}, we can put these elementary operations together to get the desired transformation on two $n$-qubit registers:
\begin{align} \label{eq_overview_eq}
    \ket{j}\ket{0}\,
    \underbrace{\overset{H}{\longrightarrow} \, \ket{j} \ket{\phi(0)} \,
    \overset{\textnormal{QFS}}{\longrightarrow} \,\ket{j}\ket{\phi(j)} \,
    \overset{\textnormal{FPE}}{\longrightarrow} \,\ket{0}\ket{\phi(j)} \,
    \overset{\textnormal{SWAP}}{\longrightarrow}}_{\overset{\textnormal{QFT}_{\textnormal{uni}}}{\longrightarrow}} \,\ket{\phi(j)} \ket{0} \, .
\end{align}
\begin{figure}[!htb]
\centering
\includegraphics[width=0.55\columnwidth]{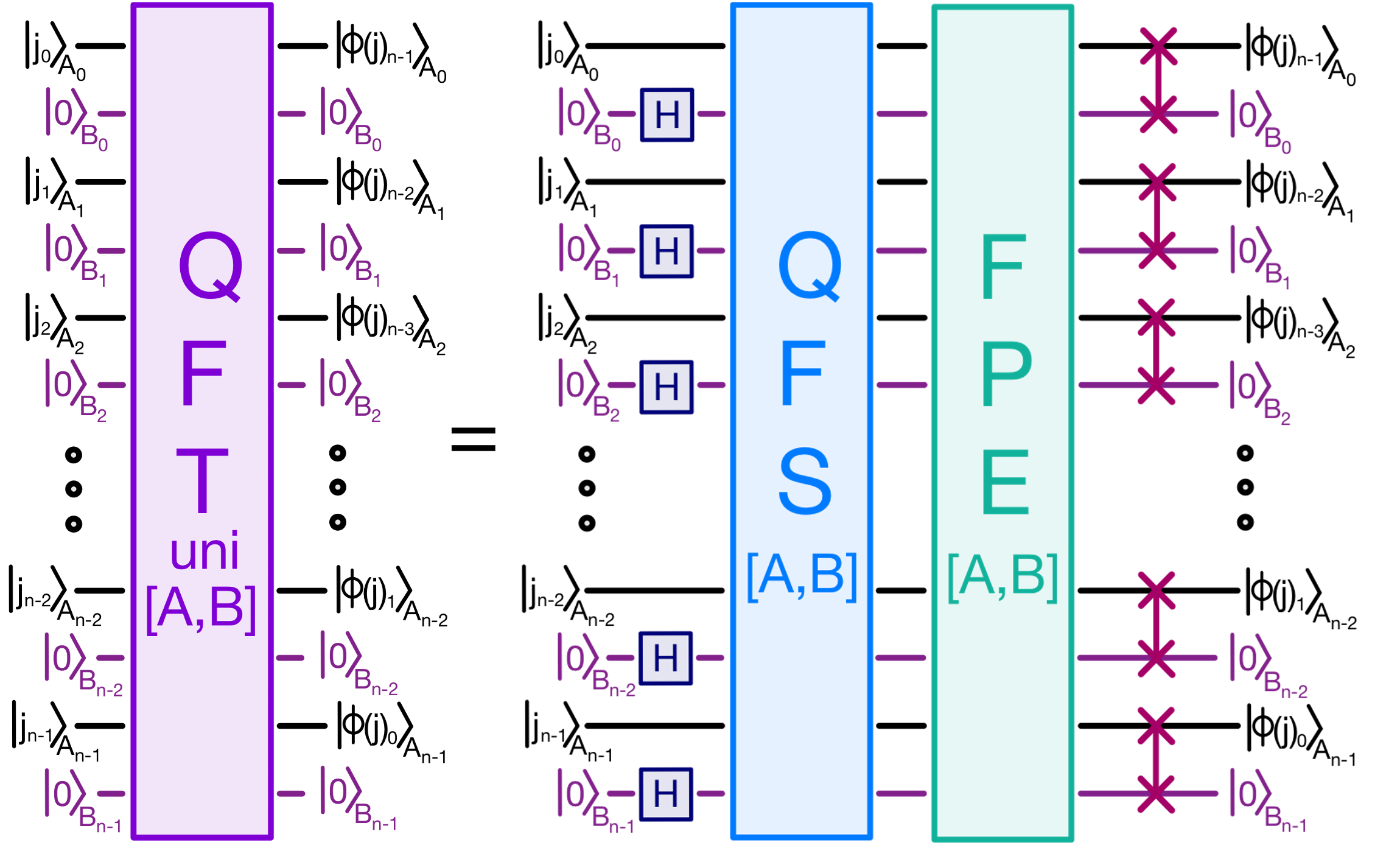}
\caption{Composite circuit for the QFT for uniformly distributed input states. Note that the two registers $A$ and $B$ are interleaved, with the output on $A$ reversed. The ideal QFT$_\mathrm{uni}$ equals the ideal QFT; the label ``uni" is only added to distinguish the two different constructions for the approximate case.} 
\label{fig:circuit_uni}
\end{figure}

However, as we see in~\cref{sec:uniform}, while the construction works for most inputs, a successful implementation is only guaranteed for inputs of the form defined in~\cref{eq_S_uni}. We can extend the result to arbitrary input states with a slightly more elaborate construction and by using an adder and showing how it can be implemented in logarithmic depth on a line using dynamic circuits.

\begin{theorem}[AQFT for arbitrary inputs] \label{thm_QFT_1D}
Let $n \in \N$ and \smash{$\frac{1}{\poly(n)}  \leq \eps < 1$}. A unitary $\mathrm{QFT}^{(\eps)}$ satisfying~\cref{eq:AQFT_def} can be implemented on a line of $4n$ qubits with nearest-neighbor connectivity with depth $\cO(\log \frac{n}{\eps^2})$.
\end{theorem}
The proof is given in~\cref{sec:general}. Again, our proof uses the ideas introduced in~\cite{hales2002}, maps the algorithm to a line of qubits using dynamic circuits and reduces its width and depth by constant factors.
The implementation of $\mathrm{QFT}^{(\eps)}$ uses a third operation:
\begin{enumerate}[(a)]
    \setcounter{enumi}{2}
    \item Quantum adder (ADD): $\ket{b}\ket{c} \rightarrow \ket{b+c}\ket{c}$ 
\end{enumerate}
together with the operations of $\mathrm{QFT}^{(\eps)}_{\mathrm{uni}}$. It can be seen that with the ADD operation above, we can also implement $\ket{\phi(j)}\ket{\phi(b+j)} \rightarrow \ket{\phi(j)}\ket{\phi(b)}$.
As illustrated in the circuit in~\cref{fig:circuit_gen}, we can put these operations together to get the desired transformation
\begin{align}
    \ket{j}_A\ket{0}_B\ket{c_1,c_2}_C
    \overset{\mathrm{ADD}_{A,C}  \mathrm{QFS}_{A,C}}&{\longrightarrow} 
    \omega_N^{jc_1}\ket{j+c_2}_A\ket{0}_B \ket{c_1,c_2}_{C}\\
    \overset{\mathrm{QFT_{uni,AB}}}&{\longrightarrow} \omega_N^{jc_1}\ket{\phi(j+c_2)}_A\ket{0}_B\ket{c_1,c_2}_{C} \\
    \overset{\mathrm{ADD}^{\dag}_{A,C} \mathrm{QFS}^{\dag}_{AC}}&{\longrightarrow} \ket{\phi(j)} \ket{0}_B \ket{c_1,c_2}_C \, , \label{eq:constr_generalinputs}
\end{align}
where register $C$ is used to enforce uniformity of the input, such that the second operation can be implemented through \smash{$\mathrm{QFT}^{(\eps)}_{\mathrm{uni}}$}. As described in more detail in~\cref{sec:general}, the register $C$ is classical and initialized with two randomly chosen numbers $c_1, c_2 \in \{0, \ldots, N{-}1\}$.
\begin{figure}[!htb]
\centering
\includegraphics[width=0.65\columnwidth]{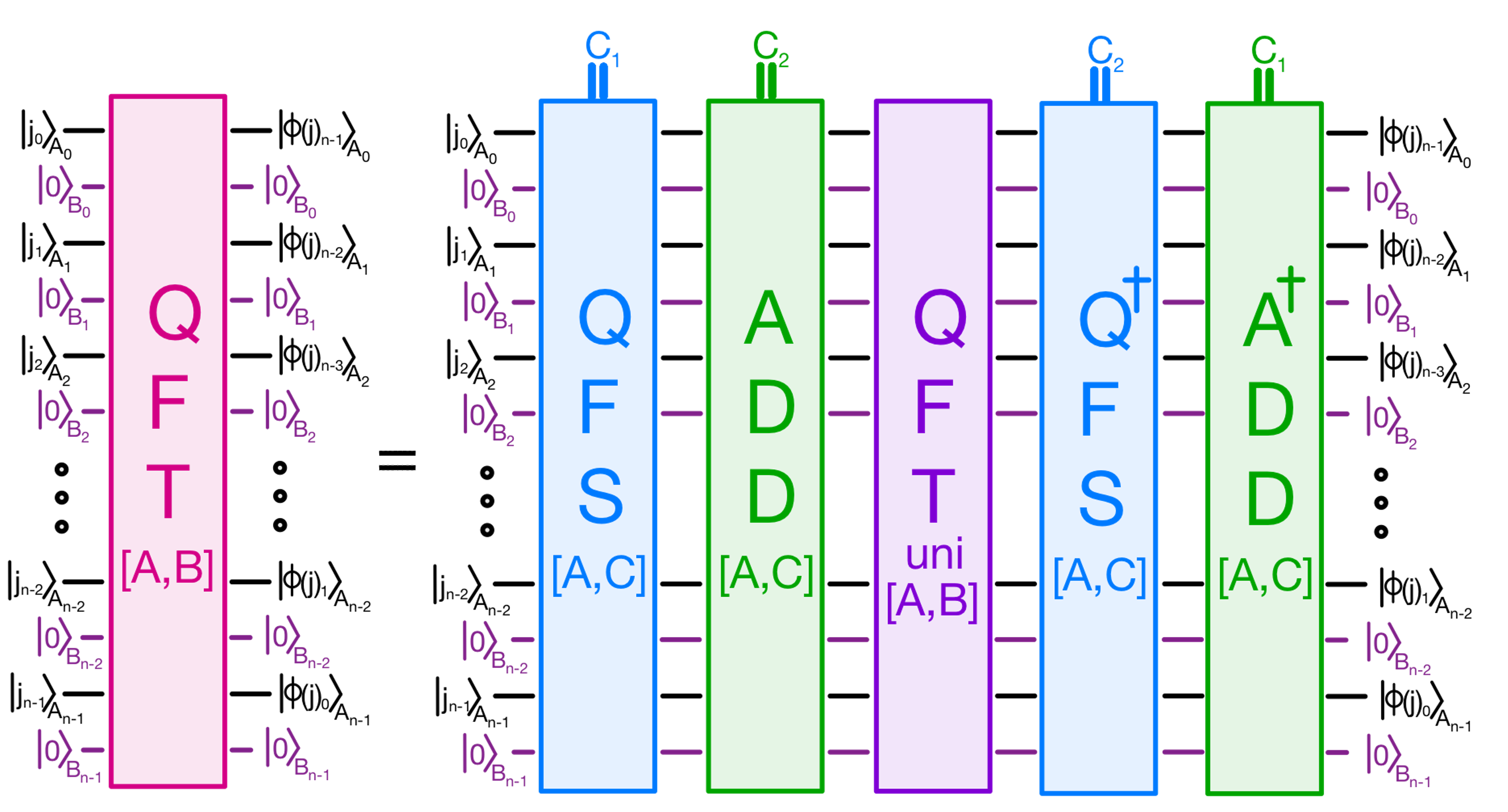}
\caption{Composite circuit for the QFT for arbitrary input states. Note that the two registers $A$ and $B$ are interleaved, with the output on $A$ reversed. The operations ADD and QFS$^\dag$ are controlled on a classical register $C$ and the adder requires two additional ancilla registers that are not shown in the figure. The ideal QFT$_\mathrm{uni}$ equals the ideal QFT; the label ``uni" is only added to distinguish the two different constructions for the approximate case.} 
\label{fig:circuit_gen}
\end{figure}

In the following, we go through (approximate) constructions of the operations QFS, FPE and ADD and show that they can be implemented on a line of qubits with corresponding depth given in~\cref{tab_QFT_depth} and width given in~\cref{tab_QFT_width}. Composing these operations as shown in~\cref{eq_overview_eq,eq:constr_generalinputs} still allows an implementation on a line and thus to prove the assertions of~\cref{lem_QFT_1D,thm_QFT_1D}.
An optimized implementation of the QFT, which reverses the order of the operations and uses classically conditioned gates for the uncomputation, as well as the implementation of QFT$^\dag$ is discussed in~\cref{app_optimized_implementation}. In~\cref{sec:shors} we discuss the implications for Shor's algorithm and some open questions.

\begin{table*}[!htb]
    \centering
    \begin{tabular}{c|ccc|cc}
        Connectivity & QFS$^{(\eps)}$ & FPE$^{(\eps)}$ & ADD & $\mathrm{QFT}^{(\eps)}_{\mathrm{uni}}$ & $\mathrm{QFT}^{(\eps)}$ \\  \hline
        All-to-all \cite{hales2002} &  $\cO(\log \frac{n}{\eps})$ & $\cO(\log \frac{n}{\eps^2})$ & $\cO(\log n)$ & $\cO(\log \frac{n}{\eps^2})$  & $\cO(\log \frac{n}{\eps^2})$  \\ 
       1D (w/o dyn.~circ.) [this work]  & $\cO(\log \frac{n}{\eps})$ & $\cO(\log \frac{n}{\eps^2})$ & $\cO(n)$ & $\cO(\log \frac{n}{\eps^2})$ & $\cO(n+\log \frac{n}{\eps^2})$ \\ 
       1D (w/ dyn.~circ.) [this work] & $\cO(\log \frac{n}{\eps})$ & $\cO(\log \frac{n}{\eps^2})$ & $\cO(\log n)$ & $\cO(\log \frac{n}{\eps^2})$  & $\cO(\log \frac{n}{\eps^2})$ \\ \hline
    \end{tabular}
    \caption{\textbf{Circuit depth} for different subroutines with error $\frac{1}{\poly(n)}  \leq \eps < 1$ comparing all-to-all and 1D connectivities. We consider \smash{$\mathrm{FPE}^{(\eps)}$} and \smash{$\mathrm{QFT}^{(\eps)}_{\mathrm{uni}}$} acting on states in \smash{$\cS_{\mathrm{uni}}^{(1)}$}.}
    \label{tab_QFT_depth}
\end{table*}

\begin{table*}[!htb]
    \centering
    \begin{tabular}{c|ccc|cc}
        Connectivity & QFS$^{(\eps)}$ & FPE$^{(\eps)}$ & ADD & $\mathrm{QFT}^{(\eps)}_{\mathrm{uni}}$ & $\mathrm{QFT}^{(\eps)}$ \\  \hline
        All-to-all \cite{hales2002}&  $2n$ & $4n$ & $4n$ & $6n$ & $8n$  \\ 
        All-to-all [this work] &  $2n$ & $2n$ & $4n$ & $2n$ & $3n$ \\ 
       1D (w/o dyn.~circ.) [this work] &  $2n$ & $2n$ & $4n$ & $2n$ & $3n$ \\ 
       1D (w/ dyn.~circ.) [this work] &  $2n$ & $2n$ & $5n$ & $2n$ & $4n$ \\  \hline
    \end{tabular}
    \caption{\textbf{Number of qubits} required for implementing the $\mathrm{QFT}^{(\eps)}$ on $2n$ qubits in logarithmic depth and with error \smash{$\frac{1}{\poly(n)} \leq \eps < 1$} for different connectivities. Note that while the general adder requires $4n$ qubits (and additional $n$ qubits for teleportation), in our implementation of $\mathrm{QFT}^{(\eps)}$ one of the numbers is classical and thus the total width decreases by $n$ qubits.}
    \label{tab_QFT_width}
\end{table*}

%%%%%%%%%%%%%%%%%%%%%%%%%%%%%%%%%%%%%%%%%%%%%%%%%%%%%%%%%%%%%%%%%%%%%%%%%%%%%%%%%%%%%%%%%%%%%%%%%%%%%%%%
%%%%%%%%%%%%%%%%%%%%%%%%%%%%%%%%%%%%%%%%%%%%%%%%%%%%%%%%%%%%%%%%%%%%%%%%%%%%%%%%%%%%%%%%%%%%%%%%%%%%%%%%
\section{Quantum Fourier transform on a line}
In this section we formally introduce the building blocks for the QFT, i.e., the QFS, FPE and ADD operations mentioned in~\cref{sec_intro}. In particular, we show how to implement approximate versions thereof in 1D circuits with nearest neighbor connectivity. To quantify the approximation error, we consider a specific distance measure that is introduced next.

\subsection{Preliminaries}\label{sec:notation}
Let $\St(\cH)$ denote the set of density  matrices on a Hilbert space $\cH$. Furthermore, for $L \in \C^{r \otimes s}$ we denote its trace norm by $\norm{L}_1:=\tr[\sqrt{L^\dagger L}]$ and its operator (or spectral) norm by $\norm{L}_{\infty}$, which is the largest singular value of $L$.
The diamond-norm distance between two unitary channels $\cU(\cdot):=U (\cdot) U^\dag$ and $\cV(\cdot):=V (\cdot) V^\dagger$ for two unitaries $U$ and $V$ on a Hilbert space $\cH_A$ can be written and bounded as 
\begin{align}
  \dist_{\diamond}(\cU,\cV)&:= \frac{1}{2} \norm{\cU-\cV}_\diamond \\ 
    &= \frac{1}{2} \max_{\rho_{AE}\in \St(\cH_{AE})}\norm{((\cU_A-\cV_A) \otimes \cI_E)(\rho_{AE})}_1 \\
    &\leq \min_{\phi \in [0,2\pi]} \norm{\ee^{\ci \phi} U - V}_{\infty} \\
    &\leq \norm{U - V}_{\infty} \\
    &= \norm{U\otimes \id_E - V\otimes \id_E }_{\infty} \\
    & = \max_{\ket{\psi}_{AE} \in \cH_{AE}} \norm{(U\otimes \id_E - V\otimes \id_E )\ket{\psi}_{AE}}_2 \\
    &=: \dist_{\cH_{AE}}(U,V) \, ,\label{eq_distances}
\end{align}
where $A\simeq E$, the second step uses~\cite[Proposition~1.6]{HKDT23}, and $\norm{\cdot}_2$ denotes the Euclidean norm. An alternative proof for the inequalities above can be found in~\cite[Lemma~12, Item~6]{AKN98}.
\Cref{eq_distances} ensures that if we can bound $\dist_{\cH_{AE}}(U,V)$, we can also bound the diamond-norm distance $\dist_{\diamond}(\cU,\cV)$.
Since for any two pure states $\ket{\omega},\ket{\tau}$ we have $\norm{\ket{\omega} - \ket{\tau}}_{2}
\leq \sqrt{2-2|\spr{\omega}{\tau}|}$ it follows that two unitaries $U$ and $V$ satisfying defined $\dist_{\cH_{AE}}(U,V) \leq \eps$ produce output states that are also close in terms of fidelity.
For a set $\cR \subseteq \cH_{AE}$ we define
\begin{align}
\dist_{\cR}(U,V):=\max_{\ket{\psi}_{AE} \in \cR} \norm{(U\otimes \id_E - V\otimes \id_E )\ket{\psi}_{AE}}_2 \, .
\end{align}

\subsection{Quantum Fourier state computation (QFS)}
\label{sec:qfs}
The exact QFS on $2n$ qubits can be implemented by the circuit depicted in~\cref{fig:qfs}, which has depth $\cO(n)$. %and size $\cO(n^2)$. 
It performs the operation
\begin{align} \label{eq_QFS}
   \ket{j}_A\ket{\phi(b)}_B  
    \overset{\mathrm{QFS}_{AB}}{\rightarrow}  \ket{j}_A \ket{\phi(b+j)}_B \,.
\end{align}
For $\eps >0$ the approximate QFS acting on two $n$-qubit registers $A$, $B$ is defined as a unitary $\mathrm{QFS}^{(\eps)}$, such that
\begin{align} 
\dist_{\cS}(\mathrm{QFS}_{AB},\mathrm{QFS}_{AB}^{(\eps)}) \leq  \eps \, , \label{eq:AQFS}
\end{align}
where $\cS$ is defined in~\cref{eq_set_S}. As shown in~\cref{app_pf_lem_QFS}, we can implement $\mathrm{QFS}_{AB}^{(\eps)}$ by neglecting small rotations, i.e.~all phase gates
\begin{align} \label{eq_def_Rk}
R_k:=\begin{pmatrix} 1 & 0\\ 0 & \ee^{2\pi\ci/2^k} \end{pmatrix}
\end{align}
for $k> k_{\max}=\cO(\log\frac{n}{\eps})$, see also~\cref{fig:qfs}.
\begin{figure}[!htb]
\centering
\includegraphics[width=1.\columnwidth]{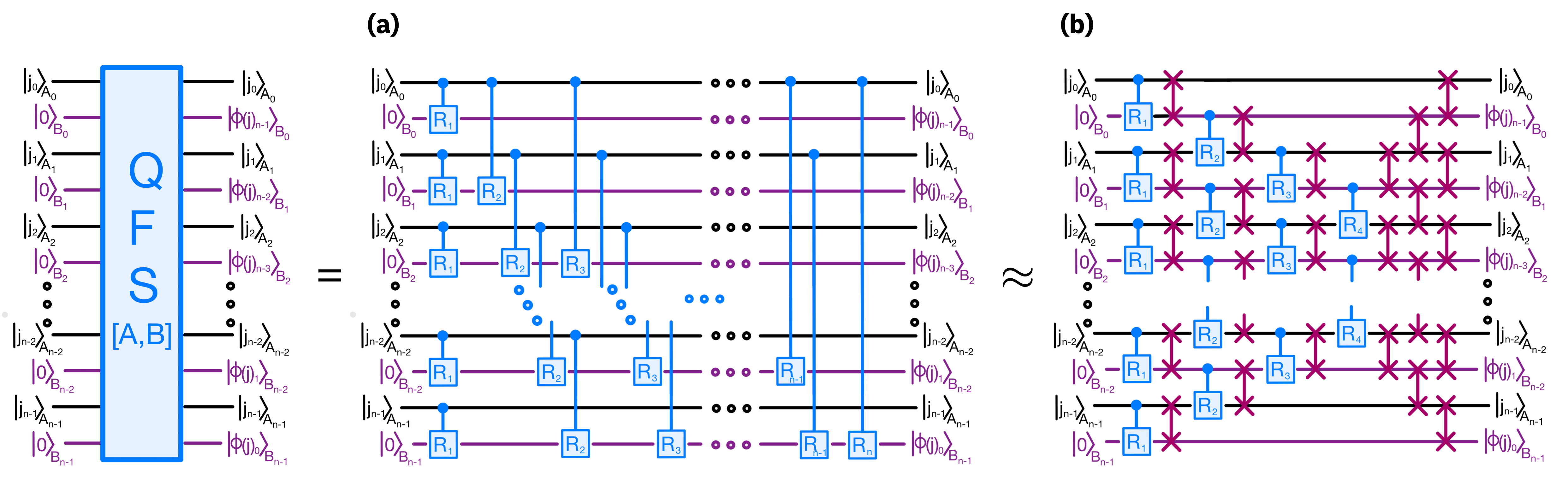}
\caption{Quantum circuit to implement the QFS defined in~\cref{eq_QFS} consisting of controlled phase gates with $R_k$ defined in~\cref{eq_def_Rk}: (a) exact QFS. (b) Approximate QFS on a 1D line, here plotted for $k_{\max}=4$.} 
\label{fig:qfs}
\end{figure}

\begin{lemma} \label{lem_QFS}
Let $n \in \mathbb{N}$ and \smash{$\frac{1}{\poly(n)}  \leq \eps < 1$}. A unitary $\mathrm{QFS}^{(\eps)}$ satisfying~\cref{eq:AQFS} can be implemented on a line of $2n$ qubits with nearest-neighbor connectivity with depth $\cO(\log \frac{n}{\eps})$.
\end{lemma}
The proof is given in~\cref{app_pf_lem_QFS}.

\subsection{Fourier phase estimation (FPE)}
\label{sec:fpe} 
The FPE, i.e., the mapping
\begin{align} \label{eq_FPE}
\ket{b}_A\ket{\phi(j)}_B \overset{\mathrm{FPE}_{AB}}{\longrightarrow} \ket{b\oplus j}_A\ket{\phi(j)}_B,
\end{align}
is the main contribution of~\cite{hales2002} that allows to reduce the width of the approximate QFT from $\cO(n\log n)$ qubits as described in~\cite{clevewatrous2000} to $\cO(n)$. The main idea, for the case $b=j$, is to estimate $\ket{j}$ by small, but exact quantum Fourier transforms that are applied in parallel on $2k$ qubits each, where $k=\cO(\log n)$. Here, we further improve the protocol presented in~\cite{hales2002} by removing two extra copies of $\ket{\phi(j)}_B$, which reduces the number of required qubits for the FPE by $2n$, and further simplifies the overall protocol of the approximate QFT. 
The circuit is illustrated in~\cref{fig:fpe}. The technical analysis is given in~\cref{app_pf_lem_FPE}.
\begin{figure}[!htb]
\centering
\includegraphics[width=0.8\columnwidth]{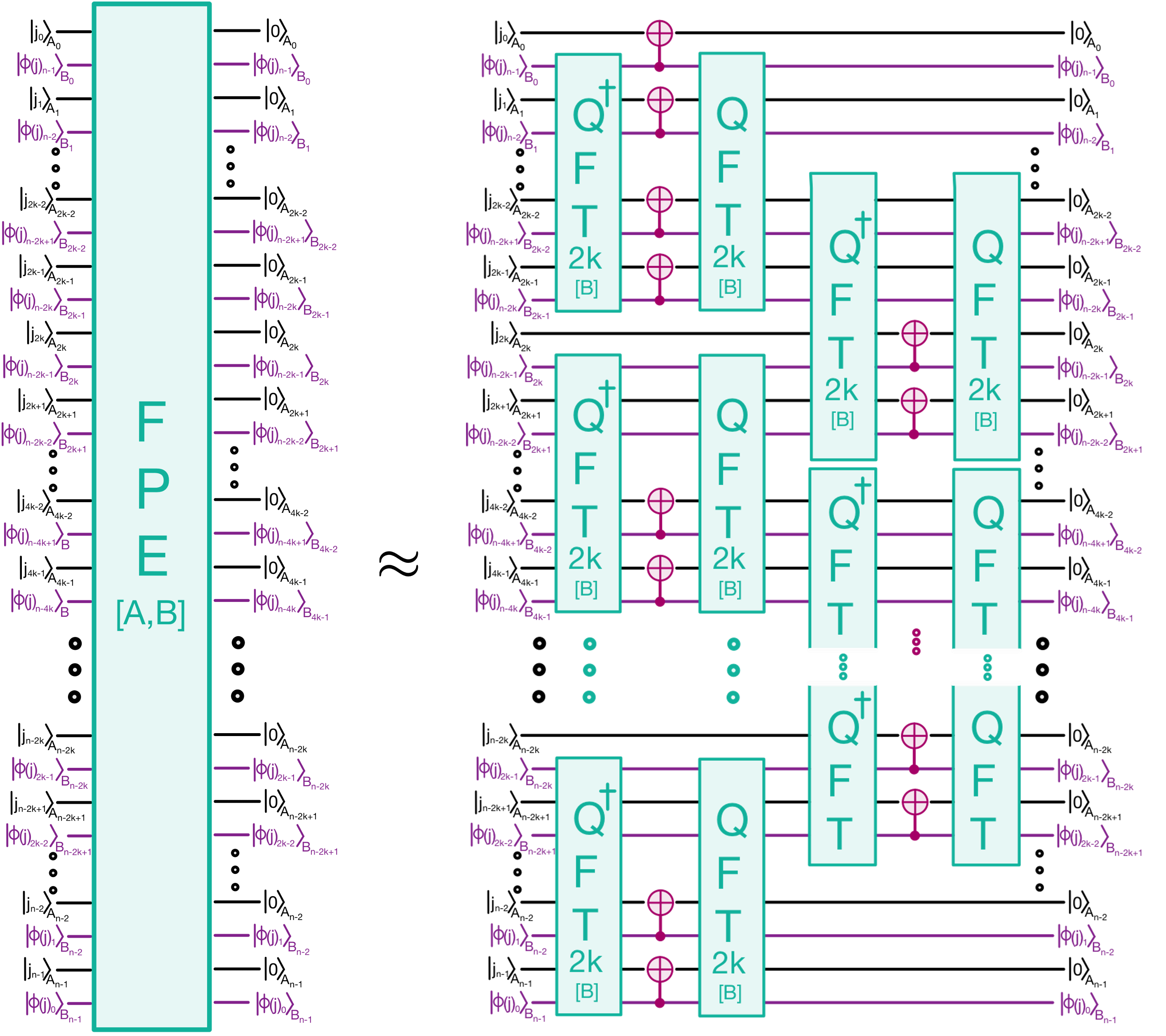}
\caption{Circuit for the approximate Fourier phase estimation defined in~\cref{eq_FPE,eq_AFPE} for $b=j$. The two registers are meshed.
} 
\label{fig:fpe}
\end{figure}

For $\eps>0$, the approximate FPE acting on two $n$-qubit registers $A$, $B$, is defined as a unitary $\mathrm{FPE}_{AB}^{(\eps)}$, such that
\begin{align} \label{eq_AFPE}
\dist_{\cT_{\mathrm{uni}}^{(p,q)}}(\mathrm{FPE}_{AB},\mathrm{FPE}^{(\eps)}_{AB}) \leq \eps \, ,
\end{align}
where
\begin{align} \label{eq_T_uni}
 \cT_{\mathrm{uni}}^{(p)}:=\Big\{\ket{\psi}_{ABE} \in \cH_{ABE}: 
 &\ket{\psi}_{ABE}=\sum_{j=0}^{N-1}\sum_{m=0}^{M-1} \beta_{j,m} \ket{j}_A \ket{\phi(j)}_B \ket{m}_{E} \ ,  \, |\sum_{m}\beta_{j,m}\beta_{\ell,m}^\ast|\leq\frac{p(n)}{N} \delta_{j,\ell} \, \forall j,\ell \Big\} \, ,
\end{align}
where $p(n)$ denotes a polynomials in $n$ with fixed degree independent of $n$ and with some environment $E$ with dimension $\dim (E) = M \geq 2^{2n}$.
\begin{lemma} \label{lem_FPE}
Let $n \in \mathbb{N}$ and \smash{$\frac{1}{\poly(n)}  \leq \eps < 1$}. 
A unitary $\mathrm{FPE}^{(\eps)}$ satisfying~\cref{eq_AFPE} can be implemented on a line of $2n$ qubits with nearest-neighbor connectivity with depth $\cO(\log \frac{n}{\eps^2})$.
\end{lemma}
The proof is given in~\cref{app_pf_lem_FPE}. While $\mathrm{FPE}^{(\eps)}$ is defined and guaranteed only for states $\ket{\psi}_{ABE} \in \cT_{\mathrm{uni}}^{(p)}$, we can see from the general construction in~\cref{sec:general} that it also works with a high probability for a random input state (see~\cref{remark:randomstate}). This probability is state-dependent though, and thus the general construction from~\cref{fig:circuit_gen} that effectively yields a state $\ket{\psi}_{ABE} \in \cT_{\mathrm{uni}}^{(p)}$ is required to guarantee a successful implementation of $\mathrm{FPE}^{(\eps)}$ for \textit{any} input state.
%%%%%%%%%%%%%%%%%%%%%%%%%%%%%%%%%%%%%%%%%%%%%%%%%%%%%%%%%%%%%%%%%%%%%%%%%%%%%%%%%%%%%%%%%%%%%%%%%%%%%%%%%%

\subsection{Quantum adder (ADD)}
\label{sec:add}
Recall that the ADD operation is the mapping
\begin{align} \label{eq_ADD_def}
\ket{b}\ket{c} \overset{\textnormal{ADD}}{\rightarrow}  \ket{b+c}\ket{c} \, .
\end{align}
Noting the following two identities
\begin{align}
\ket{\phi(j)}\ket{\phi(\ell+j)}    
=\sum_{b, c =0}^{N-1} \omega_N^{j (b+c)+\ell c}  \ket{b}\ket{c}    
\end{align}
and
\begin{align}
\ket{\phi(j)}\ket{\phi(\ell)} 
    = \sum_{b, c=0}^{N-1} \omega_N^{j b+ \ell c} \ket{b}\ket{c} 
    = \sum_{b, c=0}^{N-1} \omega_N^{j (b+c)+\ell c} \ket{b+c}\ket{c}    \, , 
\end{align}
we can see that the ADD operation defined in~\cref{eq_ADD_def} is equivalent to the mapping
\begin{align}
    \ket{\phi(j)}\ket{\phi(\ell+j)} \rightarrow \ket{\phi(j)}\ket{\phi(\ell)}.
\end{align}

Further, recall that the ADD operation is only needed for the QFT on general (i.e.~non-uniform) inputs as explained in~\cref{sec_intro}.
\begin{figure}[!htb]
\centering
\includegraphics[width=0.95\columnwidth]{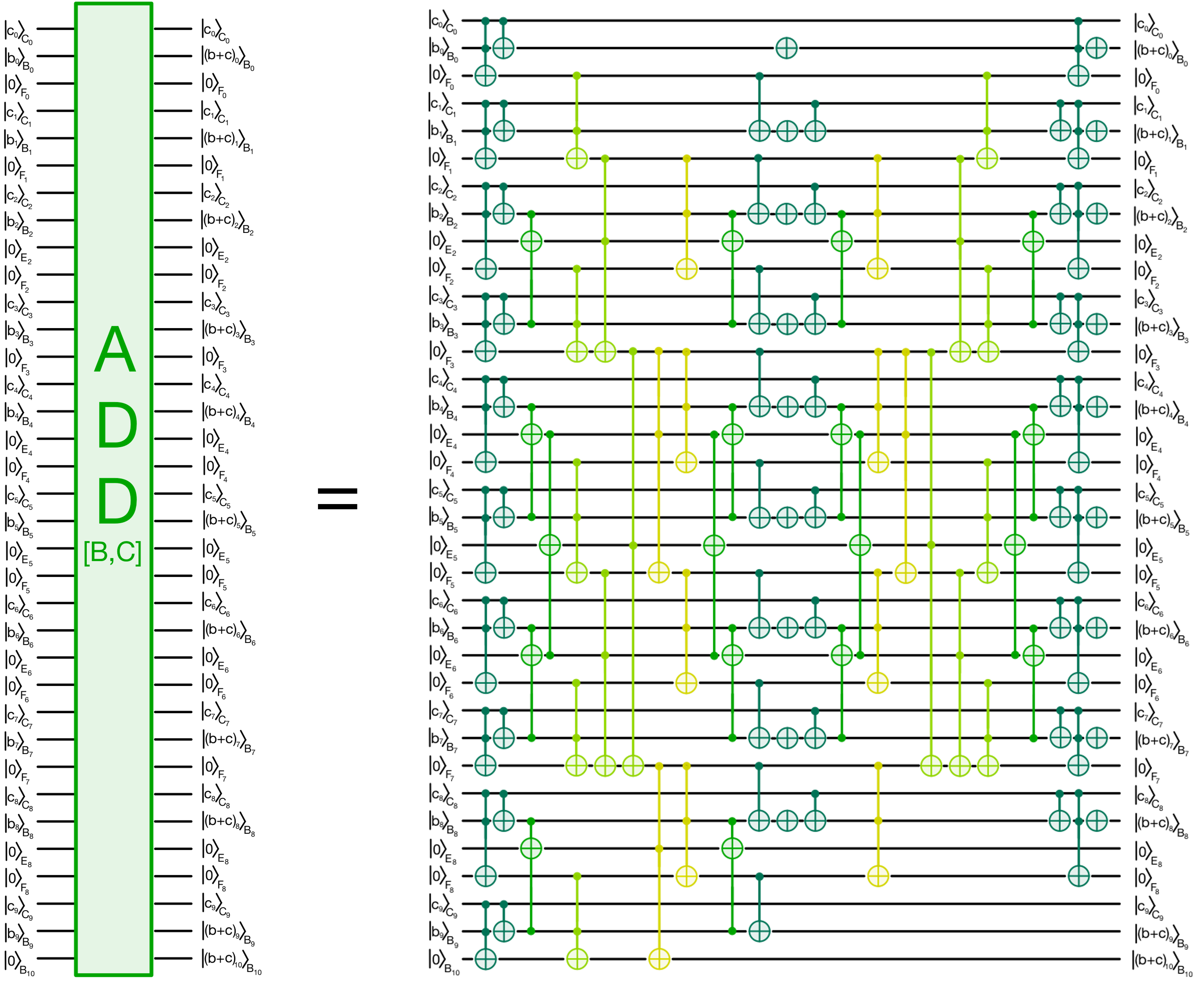}
\caption{In-place quantum carry-lookahead adder~\cite[Figure~5]{Draper06qcla}.
} 
\label{fig:qcla}
\end{figure}
\begin{lemma} \label{lem_add}
The ADD mapping defined in~\cref{eq_ADD_def} can be implemented on a line of $5n$ qubits with nearest-neighbor connectivity with depth $\cO(\log n)$.
\end{lemma}
\begin{proof}
The mapping~\cref{eq_ADD_def} can be implemented by a quantum addition circuit such as the logarithmic-depth quantum carry-lookahead adder~\cite{Draper06qcla}, see~\cref{fig:qcla}. While at first sight it seems that this adder requires all-to-all connectivity, we can exploit recent results that use dynamic circuits to implement long-range entangling gates in constant depth~\cite{baeumer2024fanout} and notice that in each of the $\cO(\log n)$ CNOT/Toffoli layers, the (long-range) gates are not overlapping and could therefore be implemented even on a 1D line in parallel. This allows to implement the desired operation
\begin{align}
    \ket{b}\ket{c} \rightarrow \ket{b+c}\ket{c}
\end{align}
with $\cO(n)$ ancilla qubits on a 1D line in depth $\cO(\log n)$. %and therefore also erasing (or copying) a Fourier basis state. %For the reverse ADD operation we simply reverse the circuit. 
Note that while the depth stays logarithmic, the size of the circuit increases from $\cO(n)$ when considering all-to-all connectivity to $\cO(n\log n)$ for a 1D line.
\end{proof}
%%%%%%%%%%%%%%%%%%%%%%%%%%%%%%%%%%%%%%%%%%%%%%%%%%%%%%%%%%%%%%%%%%%%%%%%

%%%%%%%%%%%%%%%%%%%%%%%%%%%%%%%%%%%%%%%%%%%%%%%%%%%%%%%%%%%%%%%%%%%%%%%%%%%%%%%%%%%%%%%%%%%%%%%%%%%%%%%%%%
\section{Improved implementation of the QFT} \label{app_optimized_implementation}
For a more efficient implementation, we can leverage that the ancilla registers are (ideally) uncorrelated in the end together with the principle of deferred measurement. This allows to measure some of the ancilla registers during the uncomputation and apply classically conditioned gates. Similarly to the advantage we get when implementing only classically conditioned QFS and ADD operations in the general construction of $\mathrm{QFT}^{(\eps)}$, we can reduce the depth for $\mathrm{QFT}^{(\eps)}_{\mathrm{uni}}$, first considering the implementation presented above, and then showing a further optimized way by reversing the order of the operations.

\paragraph{Forward implementation.}
Noting that the FPE essentially just unentangles and resets the first register, allows with the principle of deferred measurement to measure (and reset) it already after the QFS in the $X$-basis (i.e.~applying a Hadamard gate before the measurement). The $n$ CNOT gates of the FPE operation can then be replaced by classically controlled Z-gates applied to the second register. As these gates are classically controlled, they are not even executed if the measured qubit is in the $\ket{0}$ state, so on average we can cut down half of those operations which allows to further simplify the FPE circuit (depending on the measurement outcomes). In addition, due to the early measurements there would be a smaller idling error on the first register. 
Another possibility would be to measure the first register after the FPE has been applied and post-select on the output state $\ket{0}_A$, thereby reducing the error of the final (post-selected) state.  

\paragraph{Backward implementation.}
Let us first note that $\mathrm{QFT}^\dag \ket{j}= \sum_k \omega_N^{-jk}\ket{k}  = \ket{\phi(-j)}$ and let us define an operation $\overline{\mathrm{FPE}}$ which performs
\begin{align}
    \ket{b}\ket{\phi(-j)} \overset{\overline{\textnormal{FPE}}}{\longrightarrow} \ket{b\oplus j}\ket{\phi(-j)}.
\end{align}
Note that we can implement $\overline{\mathrm{FPE}}^{(\eps)}$ similarly to $\mathrm{FPE}^{(\eps)}$, but with all QFTs in the quantum circuit in~\cref{fig:fpe} replaced by QFT$^\dag$ and vice-versa, such that $\dist_{\overline{\cT}_{\mathrm{uni}}}(\mathrm{FPE}_{AB},\mathrm{FPE}^{(\eps)}_{AB}) \leq \eps$, where $\overline{\cT}_{\mathrm{uni}}$ is a restricted set of states, such that any $\ket{\psi}_{ABE} \in \overline{\cT}_{\mathrm{uni}}$ is a uniformly distributed input state of the form $\ket{\psi}_{ABE}=\sum_{j=0}^{N-1}\sum_{m=0}^{M-1} \beta_{j,m} \ket{j}_A \ket{\phi(-j)}_B \ket{m}_{E}$, with $|\sum_{m} \beta_{j,m}\beta_{\ell,m}^\ast| \leq\frac{p(n)}{N} \delta_{j,\ell}$, where $p(n)$ denotes a polynomial in $n$ with fixed degree independent of $n$ and some environment $E$ with dimension $\dim (E) = M \geq 2^{2n}$.

We can now see that when replacing the QFS with QFS$^\dag$ and the FPE with $\overline{\mathrm{FPE}}$ in the implementation of $\mathrm{QFT}^{(\eps)}_{\mathrm{uni}}$, we get
\begin{align} \label{eq:backwardqft}
    \ket{j}\ket{0} \!
    \overset{\textnormal{H}}{\longrightarrow} \! \ket{j}\ket{\phi(0)} \!
    \overset{\textnormal{QFS}^\dag}{\longrightarrow} \! \ket{j}\ket{\phi(-j)}\!
    \overset{\overline{\textnormal{FPE}}}{\longrightarrow} \! \ket{0}\ket{\phi(-j)} 
    \overset{\textnormal{SWAP}}{\longrightarrow} \! \ket{\phi(-j)} \ket{0} \, ,
\end{align}
i.e.~we have implemented $(\mathrm{QFT}^{(\eps)}_{\mathrm{uni}})^{\dagger}$. This directly implies that if we implement the quantum circuit in the opposite direction with all gates transposed and complex conjugated, we can again implement $\mathrm{QFT}^{(\eps)}_{\mathrm{uni}}$, but this time with a reversed order of the operations. The benefit is that in this reversed order the final operation that essentially just unentangles and resets the ancilla register is the QFS. As in the forward implementation, we can use the principle of deferred measurement to measure and reset this register and apply classically conditioned gates. As the QFS consists solely of controlled phase gates, the total phase applied to each of the unmeasured qubits can be classically determined and applied in parallel. In particular, \textit{all} $\cO(n^2)$ controlled rotations can be taken into account, so the \textit{exact} QFS can be implemented in constant depth. While this does not change the scaling of the implementation of $\mathrm{QFT}^{\eps}_{\mathrm{uni}}$, it still drastically improves the depth, mainly by decreasing the depth of the QFS from logarithmic to constant. In addition, as the exact QFS would reduce the total error by a constant factor, alternatively one could keep the total error fixed and thus decrease the depth of $\mathrm{FPE}^{(\eps)}$ as well. 
This backward implementation works also for general inputs, as the general $\mathrm{QFT}^{(\eps)}$ construction can simply be reversed as well.

Another potential optimization would be to recursively replace the small QFT operations in the construction of the FPE (see~\cref{fig:fpe}) by the approximate QFT presented here. As the inputs are already by construction uniform, the simpler $\mathrm{QFT}^{(\eps)}_{\mathrm{uni}}$ could be applied here and thus the depth of the FPE reduced to $\cO(\log \log n)$. However, this requires another $n$ qubits and a larger prefactor, and one would need to carefully analyze the total error to determine in which regime of $n$ this trade-off could become beneficial.

\section{Implications for Shor's algorithm} \label{sec:shors}
The main ingredient of Shor's algorithm~\cite{shor_algo} to factor a number $M$ is the controlled modular exponentiation, which prepares the state 
\begin{align}
    \frac{1}{2^n}\sum_{x=0}^{2^{2n}-1}\ket{x}\ket{a^x \bmod M}=\frac{1}{2^n}\sum_{z=0}^{r-1}\sum_{m=0}^{\frac{2^{2n}}{r}-1}\ket{z+mr}\ket{a^z \bmod M} \, ,\label{eq:modexp}
\end{align} 
where we used that $a^x\bmod M$ is periodic with period $r$ and for simplicity assumed that $r$ divides $2^{2n}$. This is followed by a QFT on the first register, which in practice is often implemented on a single qubit via successive measurements and feed-forward operations though~\cite{mosca1999hiddensubgroupproblemeigenvalue,zalka1998}. But as the multiplication steps for the modular exponentiation are in most implementations performed in the Fourier basis~\cite{draper2000additionquantumcomputer}, typically another $2n$ QFT and QFT$^\dag$ applications are required. Here, the input states are of the form $\sum_x \alpha_x \ket{a^x \bmod M}$, i.e., they are not uniformly distributed according to~\cref{eq_S_uni}. However, since we can easily check if the result is correct, we can always try the simplified protocol with \smash{$\mathrm{QFT}^{(\eps)}_{\mathrm{uni}}$} first and if not successful implement the general \smash{$\mathrm{QFT}^{(\eps)}$}. Note, that slightly varying the block number $k$ for the FPE changes the set of bad states $B$ for a given input state and thus may also still allow a successful implementation of \smash{$\mathrm{QFT}^{(\eps)}_{\mathrm{uni}}$}.
We leave it for future research to rigorously understand the set of input states that guarantee a successful implementation of \smash{$\mathrm{QFT}^{(\eps)}_{\mathrm{uni}}$} and in particular to understand whether the input states in Shor's algorithm are of that form for any $a$ and $M$.

Using a logarithmic depth protocol for the QFT together with a sublinear depth multiplication operation, such as the Toom-Cook quantum multiplication circuit from~\cite{kahanamokumeyer2024}, leads to a sublinear depth implementation of the controlled modular multiplication step $\ket{c}\ket{b} \longrightarrow \ket{c}\ket{a^cb\mod M}$, which is the main ingredient in Shor's algorithm. This also implies a sub-quadratic depth implementation of Shor's using only $\cO(n)$ qubits.
Unfortunately, the recursive ``PhaseProduct" operation presented in~\cite{kahanamokumeyer2024} requires an adder applied to consecutive (not interleaved) registers of order $\cO(n)$, for which a logarithmic depth implementation is only known assuming all-to-all connectivity.
Thus, finding a sublinear multiplication circuit that can be implemented on a line of qubits remains an important open question, as it would immediately imply a sub-quadratic implementation of Shor's algorithm on a line.

\subsection*{Acknowledgements}
We thank John Watrous for valuable discussions and for bringing the work~\cite{hales2002} to our attention.
%%%%%%%%%%%%%%%%%%%%%%%%%%%%%%%%%%%%%%%%%%%%%%%%%%%%%%%%%%%%%%%%%%%%%%%%%%%%%%%%%%%%%%%%%%%%%%%%%%%%%
%%%%%%%%%%%%%%%%%%%%%%%%%%%%%%%%%%%%%%%%%%%%%%%%%%%%%%%%%%%%%%%%%%%%%%%%%%%%%%%%%%%%%%%%%%%%%%%%%%%%%

%%%%%%%%%%%%%%%%%%%%%%%%%%%%%%%%%%%%%%%%%%%%%%%%%%%%%%%%%%%%%%%%%%
\appendix
\section{Proofs}
%%%%%%%%%%%%%%%%%%%%%%%%%%%%%%%%%%%%%%%%%%%%%%%%%%%%%%%%%%%%%%%%%%
\subsection{Proof of~\cref{lem_QFS}} \label{app_pf_lem_QFS}
We want to emphasize that this proof follows closely the one from~\cite{hales2002}. We repeat it here for completeness and explain a few steps in more detail. In addition, we explain how to implement everything on a line, which is novel.
By omitting small rotations, i.e.~keeping only the  
\begin{align}
R_k:=\begin{pmatrix} 1 & 0\\ 0 & \ee^{2\pi\ci/2^k} \end{pmatrix}
\end{align}
for $k\leq k_{\max}$ with $k_{\max} = \log \frac{n}{\eps'}$ for some $\eps' \in [\frac{1}{\poly(n)},\frac{1}{4\pi})$ that we will specify later in the proof.
Note that the circuit given in~\cref{fig:qfs} has depth $\cO(k)=\cO(\log \frac{n}{\eps'})$ %and size $\cO(n \log (n/\eps))$. 
and performs the operation
\begin{align}
&\ket{j}\ket{\phi(b)}  \rightarrow \ket{j}\ket{\phi_{\eps'}(b+j)} \, , \label{eq_phi_eps}
\end{align}
where for $\xi_{jk}:=\sum_{\ell} \sum_{m=0}^{n-\ell-k_{\max}-1} 2^{\ell+m-n} j_{\ell} k_m$ and a unitary $U=\sum_{p,q} \ee^{-2\pi \ci \xi_{p q}} \ket{pq}\bra{pq}$
\begin{align}
\ket{j}\ket{\phi_{\eps'}(b+j)}
&:=\frac{1}{\sqrt{N}}\sum_{k=0}^{N-1} \omega_N^{bk+\sum_{\ell} \sum_{m=n-\ell-k_{\max}}^{n-\ell-1} 2^{\ell+m} j_{\ell} k_m} \ket{j}\ket{k}\\
&= \frac{1}{\sqrt{N}}\sum_{k=0}^{N-1} \ee^{2\pi \ci (\frac{(b+j)k}{N} - \sum_{\ell} \sum_{m=0}^{n-\ell-k_{\max}-1} 2^{\ell+m-n} j_{\ell} k_m )} \ket{j}\ket{k}\nonumber\\
&=\frac{1}{\sqrt{N}}\sum_{k=0}^{N-1} \ee^{2\pi \ci (\frac{(b+j) k}{N} - \xi_{jk})} \ket{j}\ket{k}\\
&=U\frac{1}{\sqrt{N}}\sum_{k=0}^{N-1} \omega_N^{(b+j) k}\ket{j}\ket{k}\\
&=U\ket{j}\ket{\phi(b+j)} \, ,
\end{align}
where we have used the binary representation of $j=\sum_\ell 2^\ell j_\ell$ and $k=\sum_m 2^m k_m$. We can estimate
\begin{align}
    |\xi_{jk}| &\leq \sum_{\ell=0}^{n-k_{\max}-1} \sum_{m=0}^{n-\ell-k_{\max}-1} 2^{\ell+m-n}\\
    &=2^{-k_{\max}}(n-k_{\max}-1)+2^{-n} \\
    &= \frac{\eps'}{n}(n-\log\frac{n}{\eps'}-1)+2^{-n} \\
    &\leq \eps' + 2^{-n} \\
    & \leq 2 \eps' \, , \label{eq_step_EB1}
\end{align} 
where in the final step we assume that $2^{-n} \leq \eps'$.\footnote{This can be enforced by choosing $n$ sufficiently large, since as we will see later $\eps'=\cO(\eps)$.}
We next observe that
\begin{align} \label{eq_phi_eps_close}
\| \ket{j}\ket{\phi(j)}-\ket{j}\ket{\phi_{\eps'}(j)}\|_2 
&\leq \norm{\id-U}_{\infty}\\
&=\max_{j,k}|1-\ee^{-2\pi\ci \xi_{jk}}| \\
&=\max_{j,k}|2\sin(\pi \xi_{jk})| \\
&\leq |2\sin(2\pi \eps')|\label{eq:2norm_bound}
\end{align}
where the second last step uses~\cref{eq_step_EB1}.

We next choose
\begin{align} \label{eq_choice_eps_bar}
\eps' = \frac{\arcsin(\eps/2)}{2\pi} = \frac{\eps}{4\pi}+\cO(\eps^3) \,.
\end{align}
Thus, when applying the QFS to two $n$-qubit registers $A$, $B$, of an arbitrary state 
\begin{align}
    \ket{\psi}_{ABE}=\sum_{j,b=0}^{N-1} \sum_{\ell=0}^{M-1}  \alpha_{j,b,\ell} \ket{j}_A \ket{\phi(b)}_B \ket{\ell}_E \, ,
\end{align} 
that might be entangled with some environment $E$ of dimension $\dim(E)=M\geq 2^{2n}$,
\begin{align}
    \dist_{\cH_{ABE}}(\mathrm{QFS},\mathrm{QFS}^{(\eps)}) 
    &= \max_{\ket{\psi}_{ABE} \in \cH_{ABE}}\norm{(\mathrm{QFS}_{AB} \otimes \id_E)\ket{\psi}_{ABE}-(\mathrm{QFS}^{(\eps)}_{AB} \otimes \id_E)\ket{\psi}_{ABE}}_2 \\
    &=\max_{\alpha_{j,b,\ell}: \sum_{j,b,\ell}|\alpha_{j,b,\ell}|^2=1}\norm{ \sum_{j,b,\ell} \alpha_{j,b,\ell} \ket{j}_A (\ket{\phi(b+j)}_B-\ket{\phi_\eps(b+j)}_B) \ket{\ell}_E }_2\\
    &\leq \sum_{j,b,\ell}|\alpha_{j,b,\ell}|^2 \max_j \norm{\ket{j}\ket{\phi(j)}-\ket{j}\ket{\phi_{\eps'}(j)} }_2\\
    &\leq \eps \, ,
\end{align}
where we have used~\cref{eq:2norm_bound,eq_choice_eps_bar} in the last step.
Note that since $\eps'$ has the same scaling as $\eps$ (as justified in~\cref{eq_choice_eps_bar}) the circuit has depth $\cO(\log \frac{n}{\eps})$.

It remains to justify that the circuit given in~\cref{fig:qfs} can be implemented on a 1D line with $2n$ qubits. This is possible by meshing the two registers in opposite direction as indicated in~\cref{fig:qfs} and implemented in depth $\cO(\log \frac{n}{\eps})$ by $\cO(\log \frac{n}{\eps})$ layers of SWAP gates followed by controlled phase rotations (see~\cref{fig:qfs}).
\qed

%%%%%%%%%%%%%%%%%%%%%%%%%%%%%%%%%%%%%%%%%%%%%%%%%%%%%%%%%%%%%%%%

\subsection{Proof of~\cref{lem_FPE}} \label{app_pf_lem_FPE}
We want to emphasize that this proof follows closely the one from~\cite{hales2002}. We repeat it here for completeness and explain a few steps in more detail. In addition, we reduce the number of required qubits from $6n$ ($4n$ + $2n$ for the adder) to $2n$ for the case of uniformly distributed states (as defined in~\cref{eq_S_uni}), which also simplifies the overall protocol of the approximate QFT, and explain how to implement everything on a line, which is novel.
For simplicity we assume that $2k$ divides $n$ and for each copy divide the $n$ qubits into groups of $k$ qubits with indices $\mathbf{q_m}=(k(m+1)-1, ..., km)$ $\forall m \in \{0,...,n/k-1\}$. On the Fourier basis state $\ket{\phi(j)}$, we apply $\textnormal{QFT}^\dag$$\mod2^{2k}$ to each group of $2k$ qubits with indices $(\mathbf{q_{\frac{n}{k}-2m-1}, q_{\frac{n}{k}-2m-2}})$ $\forall m \in \{0,...,n/2k-1\}$, corresponding to the states
\begin{align}
    &\ket{\phi(j)}_(\mathbf{q_{\frac{n}{k}-2m-1}, q_{\frac{n}{k}-2m-2}}) = \frac{1}{2^k}\sum_{r=0}^{2^{2k}-1} \ee^{2\pi \ci /2^n (2^{k(\frac{n}{k}-2m-2)} r j)} \ket{r} \\
    \overset{\textnormal{QFT}^\dag_{2k}}{\longrightarrow}  &\frac{1}{2^{2k}} \sum_{x=0}^{2^{2k}-1}\sum_{r=0}^{2^{2k}-1}\ee^{\frac{2\pi\ci}{2^{2k}} r (2^{-2km}j-x)}\ket{x} = \sum_{x=0}^{2^{2k}-1}  \frac{1}{2^{2k}} \frac{|1-\ee^{2\pi\ci(2^{-2km}j-x)}|}{|1-\ee^{2\pi\ci(2^{-2km}j-x)/2^{2k}}|} \ket{x} =: \sum_{x=0}^{2^{2k}-1} \gamma_x^{j_{m}} \ket{x}\, , \label{eq_state_2_measure}
\end{align}
where we have used the geometric series.
If $j$ is a multiple of $2^{2km}$, i.e., $(j_{2km-1}...j_{0}) = (0...0)$, then the final state simply equals $\ket{j_{2k(m+1)-1}...j_{2km)}}=:\ket{(\mathbf{j_{q_{2m+1}}},\mathbf{j_{q_{2m}}})}$. Otherwise, we get a smeared pointmass centered at integers around the decimal $2^{-2km}j\mod 2^{2k} = j_{2k(m+1)-1}...j_{2km}.j_{2km-1}...j_{0}=:(\mathbf{j_{q_{2m+1}}},\mathbf{j_{q_{2m}}})+\delta_{j_{2m+1}}$, where we defined $\delta_{j_{2m+1}} \in (0,1)$ as the fractional part. We can bound
\begin{align}
    |\gamma_x^{j_{m}}| &\leq \frac{1}{2^{2k}|\sin (\frac{\pi}{2^{2k}}((\mathbf{j_{q_{2m+1}}},\mathbf{j_{q_{2m}}})+\delta_{j_{2m+1}}-x))|} \\
    &\leq \frac{2}{\pi} |(\mathbf{j_{q_{2m+1}}},\mathbf{j_{q_{2m}}})-x+\delta_{j_{2m+1}}|^{-1}_{2^{2k}} \\
    &\leq \frac{2}{\pi} (|(\mathbf{j_{q_{2m+1}}},\mathbf{j_{q_{2m}}})-x|_{2^{2k}}-1)^{-1}\, , \label{eq_amplitude}
\end{align}
where we have used $\sin(x)\leq x$ for $x\geq 0$, $\sin(x)\geq \frac{x}{2}$ for $x \in [0,1]$, and the notation
\begin{align}
    |x |_{N} = \left\{ \begin{array}{ll} x \mod N, & x\mod N  \leq N/2 \\
-x \mod N, & \, \textrm{else.} \\
\end{array}
\right.
\end{align}
Let us use this smeared distribution on $2k$ qubits to only estimate the first $k$ significant bits of $(\mathbf{j_{q_{2m+1}}},\mathbf{j_{q_{2m}}})$, i.e., $\mathbf{j_{q_{2m+1}}}$, as this allows to tolerate a distance of $|\mathbf{j_{q_{2m}}}|_{2^{k}}$. We can thus write
\begin{align}
    &\textnormal{QFT}^\dag_{2k} \ket{\phi(j)}_{(\mathbf{q_{\frac{n}{k}-2m-1},q_{\frac{n}{k}-2m-2}})} \nonumber \\
    &\hspace{20mm}= \sum_{x_0=0}^{2^k-1}\gamma_{(\mathbf{j_{q_{2m+1}}},x_0)}^{j_m} \ket{(\mathbf{j_{q_{2m+1}}},x_0)} + \sum_{\substack{x_1=0,\\x_1 \neq \mathbf{j_{q_{2m+1}}}}}^{2^k-1}\sum_{x_0}^{2^k-1}\gamma_{(x_1,x_0)}^{j_m}\ket{(x_1,x_0)}\\
    &\hspace{20mm}=: \sqrt{1-\eps_{j_{2m+1}}}\ket{\mathbf{j_{q_{2m+1}}}}\ket{a_{j_{2m+1}}} + \sqrt{\eps_{j_{2m+1}}}\ket{\perp_{2k}^{j_{2m+1}}} \, \label{eq:smallqftapplied},
\end{align}
where we have defined $\sqrt{1-\eps_{j_{2m+1}}}\ket{a_{j_{2m+1}}}:=\sum_{x_0=0}^{2^k-1}\gamma_{(\mathbf{j_{q_{2m+1}}},x_0)}^{j_m} \ket{x_0}$, $\ket{\perp_{2k}^{j_{2m+1}}}$ is a ``garbage" state on $2k$ qubits that is orthogonal to the rest, i.e. $\bra{\perp_{2k}^{j_{2m+1}}}(\ket{\mathbf{j_{q_{2m+1}}}}\ket{a_{j_{2m+1}}}) = 0$ and with
\begin{align}
    \eps_{j_{2m+1}} &:= \sum_{\substack{x_1=0,\\x_1 \neq \mathbf{j_{q_{2m+1}}}}}^{2^k-1}\sum_{x_0}^{2^k-1}|\gamma_{(x_1,x_0)}^{j_m}|^2\\
    \overset{\textnormal{\Cshref{eq_amplitude}}}&{\leq} 2 \sum_{x = |\mathbf{j_{q_{2m}}}|_{2^{k}}+1}^{2^{2k-1}} \frac{4}{\pi^2}(|x|_{2^{2k}}-1)^{-2}\\
    &\leq \frac{8}{\pi^2}\int_{|\mathbf{j_{q_{2m}}}|_{2^{k}}-1}^{2^{2k-1}}\frac{1}{x^2} \mathrm{d}x\\
    &\leq \frac{1}{|\mathbf{j_{q_{2m}}}|_{2^{k}}-1}\, ,
\end{align}
such that all states are normalized. 

When applied in parallel on all $n$ qubits, we get
\begin{align}
\left(\textnormal{QFT}^\dag_{2k}\right)^{\otimes \frac{n}{2k}}\ket{\phi(j)} &=\bigotimes_{m=0}^{\frac{n}{2k}-1}  \left( \sqrt{1-\eps_{j_{2m+1}}}\ket{\mathbf{j_{q_{2m+1}}}}\ket{a_{j_{2m+1}}}+\sqrt{\eps_{j_{2m+1}}}\ket{\perp_{2k}^{(j_{2m+1})}}\right) \nonumber\\
&=: \sqrt{1-\eps_{j_{\textnormal{odd}}}} \bigotimes_{m=0}^{\frac{n}{2k}-1} \ket{\mathbf{j_{q_{2m+1}}}}\ket{a_{j_{2m+1}}}+\sqrt{\eps_{j_{\textnormal{odd}}}}\ket{\perp_{n}^{(j_{\textnormal{odd}})}}, \label{eq:smallqfts}
\end{align}
where we have defined $\eps_{j_\textnormal{odd}}:=1-\prod_{m=0}^{\frac{n}{2k}-1}(1-\eps_{j_{2m+1}})$ and $\ket{\perp_{n}^{(j_{\textnormal{odd}})}}$ the garbage state on all $n$ qubits. Note that for the $\textnormal{QFT}^\dag \mod2^{2k}$ acting on $(\mathbf{q_{\frac{n}{k}-1},q_{\frac{n}{k}-2}})$, we directly get the state $\ket{(\mathbf{j_{q_{1}}},\mathbf{j_{q_{0}}})}=\ket{j_{2k-1}...j_{0}}$ as there are no fractional numbers left and we can without any error estimate the $2k$ least significant bits of $j$.

Next, we apply CNOT gates conditioned on the respective first $k$ qubits onto the first register $\ket{j}=\bigotimes_{m=0}^{\frac{n}{2k}-1} \ket{\mathbf{j_{q_{2m+1}}}} \ket{\mathbf{j_{q_{2m}}}}$, which results in
\begin{align}
\sqrt{1-\epsjodd} \bigotimes_{m=0}^{\frac{n}{2k}-1} \underbrace{\ket{\mathbf{j_{q_{2m+1}}} \oplus \mathbf{j_{q_{2m+1}}}} }_{=\ket{0}} \ket{\mathbf{j_{q_{2m}}}}\ket{\mathbf{j_{q_{2m+1}}}}\ket{a_{2m+1}}+\sqrt{\epsjodd}\ket{\perp_{2n}^{(0,j_{\textnormal{odd}})}} \, ,
\end{align}
where $\ket{\perp_{2n}^{(0,j_{\textnormal{odd}})}}$ is orthogonal to $\ket{0}$ on the first half of the first register and to $\ket{j_\textnormal{odd}}:=\bigotimes_{m=0}^{\frac{n}{2k}-1} \ket{\mathbf{j_{q_{2m+1}}}}$ on the first half of the second register.

\cref{eq:smallqfts} also implies that
\begin{align}
    \textnormal{QFT}_{2k}^{\otimes \frac{n}{2k}} \sqrt{1-\epsjodd} \bigotimes_{m=0}^{\frac{n}{2k}-1} \ket{\mathbf{j_{q_{2m+1}}}}\ket{a_{2m+1}} &= \ket{\phi(j)}-\sqrt{\epsjodd} \textnormal{QFT}^{\otimes \frac{n}{2k}} \ket{\perp_n^{(j_\textnormal{odd})}} \\
    &= (1-\epsjodd) \ket{\phi(j)} + \sqrt{\epsjodd(1-\epsjodd)}\ket{\perp_n^{(\phi(j))}} \, ,
\end{align}
where the last equation follows from normalization and $\ket{\perp_n^{\phi(j)}}$ is a garbage state orthogonal to $\ket{\phi(j)}$.
Thus, applying $\textnormal{QFT}_{2k}^{\otimes n/k}$ to the second register yields
\begin{align}
    \textnormal{FPE}_{1/2}^{(\eps)}\ket{j}\ket{\phi(j)} &=\bigotimes_{m=0}^{\frac{n}{2k}-1} \left( \!\ket{0} \ket{\mathbf{j_{q_{2m}}}} \right) ((1-\epsjodd)\ket{\phi(j)} + \sqrt{\epsjodd(1-\epsjodd)}\ket{\perp_n^{(\phi(j))}}) \nonumber \\
&\hspace{30mm}+\sqrt{\epsjodd} \textnormal{QFT}_{2k}^{\otimes n/k}\ket{\perp_{2n}^{(0,j_{\textnormal{odd}})}}\\
&=: \! (1\!-\!\epsjodd)\!\! \bigotimes_{m=0}^{\frac{n}{2k}-1}\!\! \left( \ket{0} \ket{\mathbf{j_{q_{2m}}}} \right) \ket{\phi(j)} \!+\! \sqrt{2\epsjodd\!\!-\!\epsjodd^2} \ket{{\perp}_{2n}^{(\textnormal{FPE}_{1/2},j)}}\,, \label{eq:firstsmallqft}
\end{align}
where we have defined $\textnormal{FPE}_{1/2}^{(\eps)}=\textnormal{QFT}_{2k}^{\otimes \frac{n}{2k}} \textnormal{CNOT}^{\otimes \frac{n}{2k}} (\textnormal{QFT}^\dag_{2k})^{\otimes \frac{n}{2k}}$ as the first half of the circuit implementing \smash{$\textnormal{FPE}^{(\eps)}$} and \smash{$\ket{{\perp}_{2n}^{(\textnormal{FPE}_{1/2},j)}}$} as a state orthogonal to the desired one. 

To estimate the remaining bits, we apply the same operations as before, but shifted by $k$ qubits. We start by applying $n/2k-1$ $\textnormal{QFT}^\dag$$\mod2^{2k}$ to each group of $2k$ qubits with indices $(\mathbf{q_{\frac{n}{k}-2m-2}, q_{\frac{n}{k}-2m-3}})$ $\forall m \in \{0,...,n/2k-2\}$ and as before estimate from that the first $k$ significant bits. 
Analogously to before, the state after applying $\textnormal{QFT}^\dag$$\mod2^{2k}$ shifted by $k$ qubits can be written as
\begin{align}
&(1-\epsjodd) \sqrt{1-\epsjeven} \bigotimes_{m=0}^{\frac{n}{2k}-2} \ket{\mathbf{j_{q_{2m+2}}}}\ket{a_{2m+2}}+(1-\epsjeven)\sqrt{\epsjeven}\ket{\perp_{n}^{(j_{\textnormal{even}})}} \nonumber \\
&\hspace{65mm}+\sqrt{2\epsjodd-\epsjodd^2}(\mathrm{QFT}_{2k}^\dag)^{\otimes \frac{n}{2k}} \ket{{\perp}_{2n}^{(\textnormal{FPE}_{1/2},j)}}\,, 
\end{align}
with $\epsjeven:=1-\prod_{m=0}^{\frac{n}{2k}-2}(1-\eps_{j_{2m+2}})$. 
Again, the first $k$ bits are copied into the corresponding bits of the first register via CNOT gates followed by application of $\textnormal{QFT}^\dag$$\mod2^{2k}$ to the second register. The final state is then given by
\begin{align}
    \mathrm{FPE}^{(\eps)}\! \left(\ket{j}\ket{\phi(j)}\right)&= (1-\eps_j)\ket{0}\ket{\phi(j)} + \sqrt{2\epsjodd-\epsjodd^2}\mathrm{QFT}_{2k}^{\otimes \frac{n}{2k}} \mathrm{CNOT} (\mathrm{QFT}^\dag_{2k})^{\otimes \frac{n}{2k}}  \ket{{\perp}_{2n}^{(\textnormal{FPE}_{1/2},j)}} \nonumber\\
    &\hspace{0mm}+\! (1\!-\!\epsjodd)\big(\! \sqrt{\epsjeven(1-\epsjeven)} \ket{0} \ket{{\perp'}^{\phi(j)}} \!+\!\sqrt{\epsjeven} \textnormal{QFT}_{2k}^{\otimes n/k}\ket{\perp_{2n}^{(0,j_{\textnormal{even}})}}\big) \label{eq:secondroundtrash} \\
    &=:(1-\eps_j)\ket{0}\ket{\phi(j)}+\sqrt{2\eps_j-\eps_j^2}\ket{\perp_{2n}^{(\mathrm{FPE},j)}}\, ,
    \label{eq:FPEepsperp}
\end{align}
with $\eps_j := 1-(1-\epsjodd)(1-\epsjeven)=1-\prod_{m=1}^{\frac{n}{k}-1}(1-\eps_{j_{m}})$, $\ket{\perp'_n{}^{\phi(j)}}$ another garbage state orthogonal to $\ket{\phi(j)}$ and $\ket{\perp_{2n}^{(\mathrm{FPE},j)}}$ a garbage state that is orthogonal to the desired state $\mathrm{FPE}(\ket{j}\ket{\phi(j)}) = \ket{0}\ket{\phi(j)}$.
  
The value $\eps_j$ is small for most $j$, however, for a small fraction of $j$ it is not. Let us define a set of bad values $j$, denoted $B$, by letting $j \in B$ if there exists an $m\geq 1$ such that $|\mathbf{j_{q_m}}|_{2^k} \leq 2^{k/2}$. Thus, $\forall j \notin B$, $\eps_{jm}\leq \frac{2}{(2^{k/2})}$ $\forall m$ and we can bound the error
\begin{align} \label{eq_bound_e_j}
    \eps_j = 1-\prod_{m=0}^{\frac{n}{k}-1}(1-\eps_{j_{m}}) \leq 1-(1-\frac{1}{2^{k/2}})^{\frac{n}{k}}  
    \leq \frac{n}{k2^{k/2}}\, .
\end{align}
For each $m$, the fraction of bad values in $\mathbf{j_{q_m}}$ is $2^{k/2+1}$ out of $2^k$. Thus, the fraction of $j \notin B$ is given by
\begin{align}
    \left(1-\frac{1}{2^{k/2-1}}\right)^{n/k} \geq 1-\frac{n}{k2^{k/2-1}}\, ,
\end{align}
and 
\begin{align} \label{eq_size_set_B}
|B|=\sum_{j\in B}1 \leq \frac{nN}{k2^{k/2-1}} \, .
\end{align}

Note, that the approximate FPE in~\cref{lem_FPE} is only defined for application on a restricted set $\cT_{\mathrm{uni}}^{(p)}$ of input states $\ket{\psi}_{ABE}$, such that any $\ket{\psi}_{ABE} \in \cT_{\mathrm{uni}}^{(p)}$ is of the form:
\begin{align}
    \ket{\psi}_{ABE} &=\sum_{j=0}^{N-1} \sum_{m=0}^{M-1} \beta_{j,m} \ket{j}_A \ket{\phi(j)}_B \ket{m}_{E}  \, , \label{eq:inputstate}
\end{align}
with $|\sum_{m}\beta_{j,m}\beta_{j,\ell} \leq \frac{p(n)}{N}\delta_{j,\ell}$ $\forall j,\ell$ and where $p(n)$ denotes a polynomial in $n$ with fixed degree independent of $n$, such that the input is fairly evenly distributed, and some environment $E$ with dimension $M\geq 2^{2n}$.
The uniform distribution implies that for any such input state $\ket{\psi}_{ABE}$ it holds that 
\begin{align} \label{eq_j_in_B}
\sum_{j\in B}\sum_{m}|\beta_{j,m}|^2 \leq \frac{n p(n)}{k 2^{k/2-1}} \, . 
\end{align}
Thus, we can bound the overall error 
\begin{align}
&\hspace{-5mm}\dist_{\cT_{\mathrm{uni}}^{(p)}}(\mathrm{FPE},\mathrm{FPE}^{(\eps)}) \nonumber \\ 
&=\max_{\ket{\psi}_{ABE}\in\cT_{\mathrm{uni}}^{(p)}}\norm{(\mathrm{FPE}_{AB} \otimes \id_E)\ket{\psi}_{ABE}-(\mathrm{FPE}^{(\eps)}_{AB} \otimes \id_E)\ket{\psi}_{ABE}}_2 \\
    &=\max_{\beta_{j,m}:|\sum_{m} \beta_{j,m}\beta_{\ell,m}^\ast| \leq \frac{p(n)}{N}\delta_{j,\ell}} \norm{ \sum_{j,m} \beta_{j,m} \left(\eps_j \ket{0}_A \ket{\phi(j)}_B -\sqrt{2\eps_j-\eps_j^2}\ket{\perp^{(\mathrm{FPE},j)}_{2n}}_{AB}\right)\ket{m}_E}_2 \\
    &\leq \max_{\beta_{j,m}:\sum_{m} |\beta_{j,m}|^2 \leq \frac{p(n)}{N}} \sqrt{\sum_{j,m} |\beta_{j,m}|^2 (\eps_j^2+2\eps_j-\eps_j^2) } \\
    &\leq \sqrt{ \frac{2p(n)}{N} \sum_{j} \eps_j } \\
   \overset{\textnormal{\Cshref{eq_bound_e_j}}}&{\leq} \sqrt{\frac{2p(n)}{N} \left(\sum_{j\notin B} \frac{n}{k2^{k/2}}+\sum_{j\in B} 1\right)}\\
    \overset{\textnormal{\Cshref{eq_size_set_B}}}&{\leq} \sqrt{\frac{6np(n)}{k2^{k/2}}} \, . \label{eq_to_bound}
\end{align}
For any \smash{$\frac{1}{\poly(n)} \leq \eps < 1$} we can choose $k = 2 \log (6n p(n)/\eps^2) = \cO(\log(np(n)/\eps^2))$ to bound the overall error occurring from the approximate FPE applied to the restricted set $\cT_{\mathrm{uni}}^{(p)}$, such that
\begin{align}
\dist_{\cT_{\mathrm{uni}}^{(p)}}\big(\mathrm{FPE},\mathrm{FPE}^{(\eps)}\big)
    \leq \eps.
\end{align}
As the QFT$\mod2^{2k}$ on $2k$ qubits can be exactly implemented in depth $\cO(k)$ using $\cO(k^2)$ gates, even on a 1D line, the full FPE operation consisting of $\cO(n/k)$ parallel QFTs can be implemented in depth $\cO(k)=\cO(\log \frac{np(n)}{\eps^2})$ and size $\cO(nk)=\cO(n \log \frac{np(n)}{\eps^2})$. Specifically, for $p(n)=1$ the depth is $\cO(\log \frac{n}{\eps^2})$. Note that the first $n$-qubit register with state $\ket{j}$ is meshed with the second $n$-qubit register with Fourier basis state $\ket{\phi(j)}$ in opposite direction as a result of the QFS (see also~\cref{fig:qfs}). Thus, the CNOT gates that erase state $\ket{j}$ on the first register are local gates that can be applied in parallel, as indicated in~\cref{fig:fpe}.

\qed

%%%%%%%%%%%%%%%%%%%%%%%%%%%%%%%%%%%%%%%%%%%%%%%%%%%%
\subsection{Proof of~\cref{lem_QFT_1D}} \label{sec:uniform}
In this section we compose the operations QFS and FPE to prove the assertion of~\cref{lem_QFT_1D}.
Let us recall that $\forall j$
\begin{align} 
    \ket{j}_A\ket{0}_B\overset{\id_A \otimes H_B^{\otimes n}}{\longrightarrow}  \ket{j}_A\ket{\phi(0)}_B \overset{\mathrm{QFS}}{\longrightarrow} \ket{j}_A\ket{\phi(j)}_B \overset{\mathrm{FPE}}{\longrightarrow} \ket{0}_A\ket{\phi(j)}_B \overset{\mathrm{SWAP}}{\longrightarrow} \ket{\phi(j)}_A \ket{0}_B.
\end{align}
Thus, $\mathrm{QFT} = \mathrm{SWAP} \cdot  \mathrm{FPE}  \cdot \mathrm{QFS} \cdot (\id_A \otimes H^{\otimes n}_B)$.
Let $\eps' \in (0,1)$ and $\bar \eps \in (0,1)$ be two error parameters and define 
\begin{align}
\mathrm{QFT}_\mathrm{uni}^{(\eps',\bar \eps)} := 
\mathrm{SWAP} \cdot \mathrm{FPE}^{(\bar \eps)} \cdot \mathrm{QFS}^{(\eps')} \cdot (\id_A \otimes H^{\otimes n}_B) \, ,
\end{align} 
for operations acting on $\cS_{\mathrm{uni}}^{(p)}$ as defined above in~\cref{eq_S_uni}.
Then 
\begin{align}
    &\hspace{-8mm}\dist_{\cS_{\mathrm{uni}}^{(p)}}(\mathrm{QFT},\mathrm{QFT}_{\mathrm{uni}}^{(\eps',\bar \eps)}) \nonumber \\
    &= \dist_{\cS_{\mathrm{uni}}^{(p)}}\big(\mathrm{SWAP} \cdot \mathrm{FPE} \cdot \mathrm{QFS} \cdot (\id_A \otimes H^{\otimes n}_B),\mathrm{SWAP} \cdot \mathrm{FPE}^{(\bar \eps)} \cdot \mathrm{QFS}^{(\eps')} \cdot (\id_A \otimes H^{\otimes n}_B) \big) \\
    &= \dist_{\cS_{\mathrm{uni}}^{(p)}}\big( \mathrm{FPE} \cdot \mathrm{QFS},\mathrm{FPE}^{(\bar \eps)} \cdot \mathrm{QFS}^{(\eps')}  \big) \\    
    &= \dist_{\cS_{\mathrm{uni}}^{(p)}}\big( \mathrm{FPE} \cdot \mathrm{QFS},\mathrm{FPE}^{(\bar \eps)} \cdot \mathrm{QFS}^{(\eps')}  \big) \\  
    \overset{\textnormal{triangle ineq.}}&{\leq} \dist_{\cS_{\mathrm{uni}}^{(p)}}\big( \mathrm{FPE} \cdot \mathrm{QFS},\mathrm{FPE}^{(\bar \eps)} \cdot \mathrm{QFS}  \big) + \dist_{\cS_{\mathrm{uni}}^{(p)}}\big(\mathrm{FPE}^{(\bar \eps)} \cdot \mathrm{QFS},\mathrm{FPE}^{(\bar \eps)} \cdot \mathrm{QFS}^{(\eps')}  \big)\\
    &\leq\dist_{\cT_{\mathrm{uni}}^{(p)}}\big( \mathrm{FPE},\mathrm{FPE}^{(\bar \eps)}\big) + \dist_{\cS_{\mathrm{uni}}^{(p)}}\big(\mathrm{QFS},\mathrm{QFS}^{(\eps')}  \big)\\
    \overset{\textnormal{\Cshref{lem_QFS,lem_FPE}}}&{\leq} \bar\eps +\eps' \, ,
\end{align}
where we used in the penultimate step that the QFS applied to input states $\ket{\psi}_{ABE}\in \cS_{\mathrm{uni}}^{(p)}$ yields states $\ket{\psi'}_{ABE}\in \cT_{\mathrm{uni}}^{(p)}$. The operation \smash{$\mathrm{QFT}_\mathrm{uni}^{(\eps)}$} can be implemented according to~\cref{lem_QFS,lem_FPE} in 1D with depth $\cO(\log \frac{n}{\eps'})+\cO(\log \frac{np(n)}{\bar\eps^2})$.
Choosing the error terms as $\eps'= \bar \eps  = \frac{\eps}{2}$ yields
\begin{align}
\dist_{\cS_{\mathrm{uni}}^{(p)}}\big(\mathrm{QFT},\mathrm{QFT}_{\mathrm{uni}}^{(\frac{\eps}{2},\frac{\eps}{2})}\big) \leq \eps \, ,
\end{align}
and can be implemented in depth $\cO(\log\frac{np(n)}{\eps^2})$ for input states in $\cS_{\mathrm{uni}}^{(p)}$ and specifically in depth $\cO(\log\frac{n}{\eps^2})$ for input states in $\cS_{\mathrm{uni}}^{(1)}$.
\qed
%%%%%%%%%%%%%%%%%%%%%%%%%%%%%%%%%%%%%%%%%%%%%%%%%%%%

\subsection{Proof of~\cref{thm_QFT_1D}} \label{sec:general}
As we have shown in~\cref{sec:uniform} that we can implement the unitary $\mathrm{QFT}^{(\eps)}_{\mathrm{uni}}$ which approximately performs the QFT on uniform input states with an error $\eps$. Here, we show how this result can be lifted to work for arbitrary input states by constructing a circuit that first enforces the uniformity, then applies the $\mathrm{QFT}^{(\eps)}_{\mathrm{uni}}$, and finally decodes such that the correct output state is reached.
To see how this construction (also indicated in~\cref{fig:circuit_gen}) works, consider
\begin{align}
    \ket{j}_A\ket{0}_B\ket{c_1}_{C_1}\ket{c_2}_{C_2} 
    &\overset{\mathrm{QFS}_{A,C_1}}{\longrightarrow} \omega_N^{jc_1}\ket{j}_A\ket{0}_B\ket{c_1}_{C_1}\ket{c_2}_{C_2}\\
    &\overset{\mathrm{ADD}_{A,C_2}}{\longrightarrow} \omega_N^{jc_1}\ket{j+c_2}_A\ket{0}_B\ket{c_1}_{C_1}\ket{c_2}_{C_2}\\
    &\overset{\mathrm{QFT_{A,B}}}{\longrightarrow} \omega_N^{jc_1}\ket{\phi(j+c_2)}_A\ket{0}_B\ket{c_1}_{C_1}\ket{c_2}_{C_2}\\
    &\overset{\mathrm{QFS}^\dag_{A,C_2}}{\longrightarrow} \omega_N^{jc_1}\ket{\phi(j)}_A\ket{0}_B\ket{c_1}_{C_1}\ket{c_2}_{C_2} \\
    &\overset{\mathrm{ADD}^{\dag}_{A,C_1}}{\longrightarrow}
    \ket{\phi(j)}_A\ket{0}_B\ket{c_1}_{C_1}\ket{c_2}_{C_2} \,. \label{eq:constr_generalinputs2}
\end{align}
As this holds $\forall \ket{j}_A$ and $\forall \ket{c}_C:=\ket{c_1}_{C_1}\otimes \ket{c_2}_{C_2}$, by linearity it also holds for an arbitrary superposition state $\sum_{j,c}\beta_{j,c}\ket{j}\ket{c}$.
Ensuring that \smash{$\sum_{c}|\beta_{j,c}|^2=\frac{1}{N}$} allows to apply \smash{$\mathrm{QFT}^{(\eps)}_{\mathrm{uni}}$} in the second step.

Let us consider an arbitrary input state $\ket{\psi}_{ABE}\in\cS$ and an ancillary equal superposition state \smash{$\ket{c}_C:=\frac{1}{N}\sum_{c_1,c_2=0}^{N-1}\ket{c_1}_{C_1}\ket{c_2}_{C_2}$} on $2n$ qubits. Then, we have
\begin{align}
    \ket{\psi}_{ABE}\ket{c}_C 
    &= \frac{1}{N} \sum_{j,m,c} \alpha_{j,m} \ket{j}_{A} \ket{0}_B \ket{m}_{E}  \ket{c}_{C} \\
    &\overset{\mathrm{ADD}_{A,C_2}\cdot \, \mathrm{QFS}_{A,C_1} }{\longrightarrow}  \frac{1}{N} \sum_{j,m, c} \alpha_{j,m} \omega_N^{jc_1}\ket{j+c_2}_{A} \ket{0}_B \ket{m}_{E}  \ket{c}_{C} \\
    &= \frac{1}{N} \sum_{j,m, c} \alpha_{j-c_2,m} \omega_N^{(j-c_2)c_1} \ket{j}_{A} \ket{0}_B \ket{m}_{E}  \ket{c}_{C} \\
    & =\sum_{j,m'} \beta_{j,m'} \ket{j}_A\ket{0}_B\ket{m'}_{E'}\\
    &=: \ket{\psi'}_{ABE'}\, ,
\end{align}
with $\ket{m'}_{E'} :=\ket{m}_E\ket{c}_C$ and $\beta_{j,m'}:=\frac{\alpha_{j-c_2,m}}{N}\omega_N^{(j-c_2)c_1}$, such that 
\begin{align}
    \sum_{m'}\beta_{j,m'}\beta_{\ell,m'}^\ast = \frac{1}{N^2}\sum_{m,c_1}  \alpha_{j-c_2,m}\alpha_{\ell-c_2,m}^\ast \underbrace{\sum_{c_2} \omega_N^{(j-\ell)c_2}}_{\delta_{j,\ell}N}  = \frac{\delta_{j,\ell}}{N}\sum_{m,c_1} |\alpha_{c_1,m}|^2 = \frac{\delta_{j,\ell}}{N} \,.
\end{align}
Thus, for any input state $\ket{\psi}_{ABE}\in\cS$, the construction yields a state $\ket{\psi'}_{ABE'}\in\cS_{\mathrm{uni}}^{(1)}$ which allows the application of $\mathrm{QFT}^{(\eps)}_{\mathrm{uni}}$. Following the construction described in~\cref{eq:constr_generalinputs2} then yields for arbitrary states
\begin{align}
    \ket{\psi}_{ABE}\ket{c}_C \overset{\mathrm{ADD}_{A,C_2}\cdot \, \mathrm{QFS}_{A,C_1}}&{\longrightarrow}  \sum_{j,m'} \beta_{j,m'} \ket{j}_A\ket{0}_B\ket{m'}_{E'} \\
    \overset{\mathrm{QFT}_{AB}}&{\longrightarrow} \! \sum_{j,m'}\! \beta_{j,m'}\! \ket{\phi(j)}_A\ket{0}_B\ket{m'}_{E'} \!=\! \frac{1}{N}\! \sum_{j,m, c}\! \alpha_{j,m} \omega^{jc_1} \ket{\phi(j+c_2)}_{A} \ket{0}_B \ket{m}_{E} \ket{c}_{C} \\
    \overset{\mathrm{ADD}^\dag_{A,C_1}\cdot \, \mathrm{QFS}^\dag_{A,C_2}}
    &{\longrightarrow}  \frac{1}{N} \sum_{j,m, c} \alpha_{j,m} \ket{\phi(j)}_{A} \ket{0}_B \ket{m}_{E}  \ket{c}_{C} = \mathrm{QFT}_{AB}\ket{\psi}_{ABE} \ket{c}_{C} \, .
\end{align}
Note, that the ancilla register $C$ is uncorrelated in the end, which allows to be measured, and that all operations that involve register $C$ are only controlled on it. Thus, we can use the principle of deferred measurement to measure already in the very beginning and apply only classically conditioned gates. This, on the other hand, can be replaced by just classically drawing two random numbers $c_1$, $c_2$, adding them to the initial state and phase via classically controlled ADD and QFS operations, and later subtracting the corresponding states and phases again via the classically controlled QFS$^\dag$ and ADD$^\dag$. This procedure yields the exact same statistics as using an equally distributed state $\ket{c}_{C}$ and allows to implement the classically controlled $\mathrm{QFS}$ and $\mathrm{QFS}^\dag$ without any approximation error, as it simply corresponds to classically controlled single qubit phase rotations. It further substantially reduces the depth as compared to the regular QFS and ADD operations on $2n$ qubits, as all two-qubit gates controlled on $C$ can be replaced by classically controlled single-qubit rotations, which have no connectivity constraints anymore.

For $\mathrm{QFT}^{(\eps)}_{AB}=\mathrm{ADD}_{A,C_1}^\dag  \mathrm{QFS}_{A,C_2}^\dag \cdot \mathrm{QFT}^{(\eps)}_{\mathrm{uni}} \cdot \mathrm{ADD}_{A,C_2} \mathrm{QFS}_{A,C_1}$ we can now bound
\begin{align}
    \dist_{\cS}(\mathrm{QFT}_{AB},\mathrm{QFT}_{AB}^{(\eps)}) 
    &=\max_{\ket{\psi}_{ABE}\in \cS}\norm{\big((\mathrm{QFT}_{AB}-\mathrm{QFT}_{AB}^{(\eps)}) \otimes \id_E\big)\ket{\psi}_{ABE}}_2 \\
    &=\max_{\ket{\psi}_{ABE}\in \cS}\!\norm{ \Big( \mathrm{ADD}^\dag  \mathrm{QFS}^\dag (\mathrm{QFT}\!-\!\mathrm{QFT}_{\mathrm{uni}}^{(\eps)}) \mathrm{ADD} \, \mathrm{QFS} \big) \otimes \id_E\big)\ket{\psi}_{ABE}}_2 \\
    &\leq \max_{\ket{\psi'}_{ABE'}\in \cS_{\mathrm{uni}}^{(1)}} \norm{ (\mathrm{QFT}-\mathrm{QFT}_{\mathrm{uni}}^{(\eps)}) \otimes \id_E\big)\ket{\psi'}_{ABE'}}_2\\
    &= \dist_{\cS_{\mathrm{uni}}^{(1)}}(\mathrm{QFT},\mathrm{QFT}_{\mathrm{uni}}^{(\eps)})\\
    &\leq \eps \, ,
\end{align}
where the final step uses~\cref{lem_QFT_1D}.
The circuits required to perform the approximate QFT with error $\eps$ as explained above are visualized in~\cref{fig:circuit_gen} and can be implemented according to~\cref{lem_QFS,lem_FPE,lem_add} in 1D with depth
\begin{align}
\underbrace{\cO(1)}_{\mathrm{QFS}} + \underbrace{\cO(\log n)}_{\mathrm{ADD}} + \underbrace{\cO(\log \frac{n}{\eps^2})}_{\mathrm{QFT}_{\mathrm{uni}}}  + \underbrace{\cO(1)}_{\mathrm{QFS}^\dag}+ \underbrace{\cO(\log n)}_{\mathrm{ADD}^\dag}
=\cO\Big(\log \frac{n}{\eps^2}\Big) \, .
\end{align}
Note that $\mathrm{QFT}_{\mathrm{uni}}$ is applied to states in $\cS_{\mathrm{uni}}^{(1)}$, thus its depth is bounded by $\cO(\log \frac{n}{\eps^2})$ and as the QFS is classically controlled here, its depth is constant, as explained above.
\qed

\begin{remark} \label{remark:randomstate}
    Note, that the construction described above adds a random state and phase to the input state before applying $\mathrm{QFT}_{\mathrm{uni}}$ and shows that the expectation value of the error can be bounded. This implies, that for most states and phases the error can be bounded with the same scaling. In particular, also for $c_1=c_2=0$, i.e., $\mathrm{QFT}_{\mathrm{uni}}$ is likely to work for a random input state without adding a state and phase. However, the success probability is state-dependent, and thus it cannot be guaranteed to work for an arbitrary input state and we do not know if it works without the general construction for relevant states. As the set of bad states $B$ depends on the block size $k$ though, varying $k$ may also remove a certain input state from $B$ and thus allow for a successful implementation.
\end{remark}

\begin{remark}
Unlike in~\cite{hales2002}, the presented $\mathrm{QFT}$ acts only on two registers of $n$ qubits each, instead of four. However, for general input states that require the construction described in~\cref{eq:constr_generalinputs2}, we require a quantum adder. As shown in~\cref{lem_add}, this can be implemented on a line in logarithmic depth, but it requires in general $4n$ qubits. In our specific case, where a classical number is added to a quantum register, we can apply classically conditioned operations and reduce the number of qubits for the adder to $3n$. However, given the local connectivity, yet another $n$ qubits are required for teleporting the long-range entangling gates. As the adder is only applied when register $B$ is in state $\ket{0}_B$, these $n$ qubits can also be used as ancilla qubits. Thus, for general inputs we require $4n$ qubits for implementing the AQFT on $n$ qubits on a line. 
To adhere to the connectivity constraints for the different operations, all four registers need to be ``meshed", meaning that the registers the qubits are part of are alternating. As this implies that two qubits of the same register are separated by three qubits of the other registers, these qubits often need to be ``skipped", which increases the depth and size of the circuit only by a constant factor though. Still, this makes the implementation of $\mathrm{QFT}^{(\eps)}_{\mathrm{uni}}$ substantially cheaper.
The full circuit is sketched in~\cref{fig:circuit_gen}. 
\end{remark}

%%%%%%%%%%%%%%%%%%%%%%%%%%%%%%%%%%%%%%%%%%%%%%%%%%%%

\bibliographystyle{arxiv_no_month}
\bibliography{bibliography}

\end{document}